\RequirePackage{import}
\documentclass{article}
\usepackage{iclr2027_conference,times}

\usepackage{amsmath,amsfonts,bm}

\def\eqref#1{equation~\ref{#1}}

\def\1{\bm{1}}

\DeclareMathAlphabet{\mathsfit}{\encodingdefault}{\sfdefault}{m}{sl}
\SetMathAlphabet{\mathsfit}{bold}{\encodingdefault}{\sfdefault}{bx}{n}

\DeclareMathOperator*{\argmax}{arg\,max}
\DeclareMathOperator*{\argmin}{arg\,min}

\usepackage{hyperref}
\usepackage{url}
\usepackage{xspace}
\usepackage{graphicx}
\usepackage{wrapfig}
\usepackage{flafter}
\usepackage{placeins}
\usepackage{algorithm}
\usepackage{algpseudocode}
\usepackage{enumitem}
\usepackage{booktabs}
\usepackage{multirow}
\usepackage{siunitx}
\usepackage{pifont}
\usepackage{tabularx}
\usepackage{array}
\usepackage{amssymb}
\usepackage{amsthm}
\newtheorem{proposition}{Proposition}[section]
\theoremstyle{remark}
\newtheorem*{remark}{Remark}
\usepackage[table]{xcolor}
\newcommand{\shortname}{AlphaPADI}

\newcommand{\highlight}[1]{
  \par\smallskip
  \noindent\textbf{#1}\ 
}

\title{AlphaPADI: Formulaic Alpha Discovery Via Pool-Aware Hierarchical Discrete Diffusion}

\iclrfinalcopy
\begin{document}

\author{Yanzheng Jin\textsuperscript{1},
Pengyang Shao\textsuperscript{1},
Yunshan Ma\textsuperscript{2},
Haowen Pan\textsuperscript{3}\\
\textbf{Naixin Zhai\textsuperscript{3},
Chen-Hui Song\textsuperscript{4},
Fei Shen\textsuperscript{1},
Kenji Kawaguchi\textsuperscript{1}}\\
\textsuperscript{1}National University of Singapore,
\textsuperscript{2}Singapore Management University\\
\textsuperscript{3}University of Science and Technology of China,
\textsuperscript{4}Tsinghua University
}

\maketitle
\lhead{}

\begin{abstract}
Formulaic alpha discovery seeks symbolic expressions that predict
cross-sectional asset returns. In deployment, multiple formulas are
combined into an alpha pool, where each formula is valued through the
complementary information it contributes to joint predictive
performance. 
While Reinforcement Learning and Generative Flow Networks have emerged as promising approaches for generating formulaic alphas, existing frameworks face three related
challenges. First, generating formulas individually leaves pool context
and inter-formula complementarity outside the generative state. Second,
formula-wise generation lacks a unified mechanism for preserving and
revising structures at different levels.
Third, pool-level rewards jointly reflect predictive performance and
redundancy but cannot be differentiated directly through symbolic
evaluation to train the generator. To overcome these challenges, we introduce
\textbf{\shortname{}} (Formulaic \textbf{\underline{Alpha}} Discovery
via \textbf{\underline{P}}ool-\textbf{\underline{A}}ware Hierarchical
Discrete \textbf{\underline{DI}}ffusion), a novel framework built around
three components: (1) grammar-constrained buffer initialization that
constructs syntactically valid pool candidates, (2) pool-aware
hierarchical diffusion that reconstructs complete pools at multiple
structural scales under the current pool context, and (3) reward-guided
pool refinement that evaluates joint predictive performance and inner
diversity, updates the elite buffer, and trains the reverse model through
reconstruction and preference learning. Empirical results on the Chinese and U.S. stock markets demonstrate that \shortname{}
outperforms the evaluated baselines in both predictive and portfolio
performance, thereby validating pool-aware generation as an effective
framework for automated alpha discovery.
\end{abstract}

\begingroup
\setlength{\abovecaptionskip}{4pt}
\setlength{\textfloatsep}{8pt plus 2pt minus 2pt}
\setlength{\floatsep}{8pt plus 2pt minus 2pt}
\setlength{\intextsep}{8pt plus 2pt minus 2pt}
\section{Introduction}
\label{sec:intro}

Formulaic alpha discovery aims to mine interpretable symbolic signals
that predict cross-sectional asset returns and support systematic
portfolio construction. These predictive signals are commonly referred
to as \emph{alphas}~\citep{kakushadze2016101}. A formulaic alpha represents such a signal as a symbolic expression over market variables, such as
$\operatorname{TsMean}(\mathrm{Close},5)-\mathrm{Open}$.
This explicit structure makes the signal interpretable and its computation auditable~\citep{yang2020qlib}. Discovering effective formulas
requires searching the combinatorial space of features, operators, and
temporal windows~\citep{zhang2020autoalpha}. In practice, multiple alphas are deployed as a pool
whose joint prediction drives asset ranking and portfolio construction~\citep{yu2023generating}.
The practical goal is therefore not merely to find strong individual
formulas, but to identify formulas that contribute complementary
predictive information when combined.

Traditionally, formulaic alpha discovery has followed a manual,
hypothesis-driven process. Researchers formulate financial or economic
hypotheses, translate them into symbolic expressions, and assess their
predictive power through statistical tests and backtesting~\citep{fama1993common}. This process
produces interpretable signals grounded in domain knowledge, but the
number of hypotheses that can be examined is limited by human effort and
prior experience. As financial data become increasingly high-dimensional,
the space of nonlinear interactions among variables, operators, and
temporal windows grows beyond the reach of manual enumeration~\citep{gu2020empirical}. Manual
discovery also makes it difficult to search systematically for
complementary formulas instead of repeatedly expressing similar market
effects. These limitations have motivated automated alpha mining, in
which learning algorithms search large formula spaces at scale.

Automated alpha mining explores formula space through evolutionary
search, reinforcement learning, and generative modeling~\citep{shi2025alphaforge}. Genetic
algorithms evolve formulas through mutation, crossover, and fitness-based
selection~\citep{koza1994genetic,zhang2020autoalpha,chen2021empirical}. Reinforcement
learning formulates expression construction as a sequential decision
process and can score a formula by its contribution to a downstream
combination model~\citep{yu2023generating,zhu2025alphaqcm,zhao2025quantfactor}. Generative flow networks instead learn to sample complete formulas with probabilities
proportional to their rewards, encouraging diverse coverage of
high-reward expressions~\citep{bengio2021flow,chen2026alphasage}. Despite their different search mechanisms, these approaches typically retain the individual
formula as the generative object. This formula-wise paradigm creates
three related challenges. First, pool context and inter-formula
complementarity remain outside the generative state, so they cannot
determine which formula structure should be modified (e.g., \textit{replacing a moving-average subtree when the pool already contains similar signals}). Second, sequential
neural generators lack a unified mechanism for preserving useful
structures while revising them at the leaf, subtree, or whole-formula
level~\citep{kusner2017grammar,kapur2025diffusion}. Third, pool-level rewards obtained through symbolic execution and
joint fitting provide no direct gradients for learning structural
revisions guided by pool context~\citep{black2024training}. 
These challenges motivate incorporating complete pools into both
structural generation and learning.

To overcome these challenges, we introduce \textbf{\shortname{}}
(Formulaic \textbf{\underline{Alpha}} Discovery via
\textbf{\underline{P}}ool-\textbf{\underline{A}}ware Hierarchical
Discrete \textbf{\underline{DI}}ffusion), a framework that treats the
complete alpha pool as the common unit of generation, evaluation, and
learning. First, grammar-constrained buffer initialization constructs
syntactically valid pool candidates. Second, pool-aware hierarchical
diffusion~\citep{austin2021structured,sahoo2024simple} reconstructs complete pools at the leaf, subtree, and alpha
levels under the current pool context, allowing pool members to guide
structural revisions. Third, reward-guided pool refinement evaluates
joint predictive performance and inner diversity, updates an elite
buffer, and trains the reverse model through reconstruction and
preference learning~\citep{wallace2024diffusion}. Together, these components form a closed loop in
which pool-level feedback shapes subsequent candidate generation.
Our main contributions are summarized as follows:
\begin{itemize}[leftmargin=*]
    \item We reformulate alpha discovery as joint pool search and make pool
    context part of the generative state. This formulation allows
    inter-formula complementarity to govern structural revisions rather
    than serving only as an evaluation signal.
    \item We propose \textbf{\shortname{}}, which combines grammar-constrained
    buffer initialization, pool-aware hierarchical diffusion, and
    reward-guided pool refinement. Reconstruction and pool-level
    preference learning jointly preserve valid structures and steer
    generation toward higher-reward pools.
    \item Extensive experiments across the Chinese and U.S. stock markets demonstrate the effectiveness of \textbf{\shortname{}} in discovering high-quality alpha pools. Remarkably, \textbf{\shortname{}} ranks first on every reported predictive metric while achieving the highest annualized return and Sharpe ratio, as well as the lowest maximum drawdown, across all three datasets.
\end{itemize}

\section{Related Work}
\label{sec:related}

\subsection{Formulaic Alpha Discovery in Quantitative Trading}
A formulaic alpha maps historical market variables to a cross-sectional score, whose predictive quality is commonly measured by its information coefficient (IC) with subsequent returns~\citep{kakushadze2016101,yu2023generating}. Alphas are often deployed as a linearly combined pool in which correlated members add little information, so useful alphas must be both predictive and complementary. Evolutionary methods evolve expression trees through mutation, crossover, and selection~\citep{koza1994genetic,cui2021alphaevolve,zhang2020autoalpha}, but they can be inefficient and converge prematurely. Reinforcement learning generates formulas token by token. AlphaGen rewards each formula by its marginal contribution to the combined pool~\citep{yu2023generating}, and later methods improve its policy optimization~\citep{zhao2025quantfactor} and value estimation~\citep{zhu2025alphaqcm}. AlphaForge adds dynamic factor combination~\citep{shi2025alphaforge}, AlphaSAGE samples structurally diverse formulas over abstract syntax trees (ASTs) with a generative flow network~\citep{chen2026alphasage}, and language-model methods refine formulas with backtest feedback~\citep{tang2025alphaagent,shi2025navigating}.

Across these methods, the pool enters mainly through marginal-contribution rewards, diversity terms, or adaptive combination, while the generator still constructs one formula at a time. These observations motivate frameworks that use the pool to guide generation as well as evaluation. Our design follows this principle: it revises formula subtrees conditioned on the other pool members during generation, and it evaluates and learns from each candidate pool as a whole through its combined predictive performance and internal redundancy.

\subsection{Discrete Diffusion and Preference Optimization}
Discrete diffusion models reverse corruption processes over categorical objects. D3PM introduces absorbing-state transitions for discrete sequences~\citep{austin2021structured}. MDLM and MD4 simplify masked diffusion training~\citep{sahoo2024simple}, SEDD learns discrete diffusion by estimating ratios of the data distribution~\citep{lou2023discrete}, and LLaDA scales masked diffusion to large language models~\citep{nie2026large}. Structured variants corrupt categorical node and edge attributes of graphs, as in DiGress~\citep{vignac2022digress}, or mutate program syntax trees while preserving grammatical validity~\citep{kapur2025diffusion}. These formulations corrupt and reconstruct a single sequence, graph, or program. \shortname{} instead absorbs complete AST subtrees, from individual leaves to whole formulas, and reconstructs each target conditioned on the other members of the alpha pool.

To align diffusion models with downstream objectives, reinforcement learning fine-tunes denoising trajectories with policy gradients~\citep{black2024training}, and discrete counterparts backpropagate differentiable rewards through relaxed sampling~\citep{wang2025fine}. Direct preference optimization (DPO) converts pairwise comparisons into a reference-relative likelihood objective~\citep{rafailov2023direct}. Diffusion-DPO applies this principle with variational likelihood estimates~\citep{wallace2024diffusion}. Discrete extensions apply it to structured binary sequence generation~\citep{borso2025preference} and to masked diffusion language models, where sharing noise between the policy and reference estimates reduces variance~\citep{zhu2026llada}. Our preference signal compares complete pools by combined predictive performance and internal diversity. We evaluate each pair under matched corruption and score both pools with the reconstruction objective used for training. This connects a non-differentiable pool-level financial score to the reverse generator without per-formula rewards.

\section{Preliminary}
\label{sec:preliminary}
\highlight{Pool-aware alpha discovery.}
Pool-reward reinforcement learning~\citep{yu2023generating} can formulate individual-formula search through the expected marginal contribution to the current pool:
\begin{equation}
\small
f_{\mathrm{new}}^{*}
=\arg\max_{f\in\mathcal X}
\mathbb E\!\left[R_{\mathrm{marg}}(f\mid F)\right],
\quad
R_{\mathrm{marg}}(f\mid F)=U(F\uplus[f])-U(F),
\label{eq:rl_generation}
\end{equation}
where $\mathcal X$ is the formula space, $F$ is the current pool,
$\uplus$ inserts a formula into the pool, and $U(F)$ denotes the utility
of a jointly fitted pool. A formula-level generative flow network~\citep{chen2026alphasage} learns a reward-proportional distribution over individual formulas: $P_{\phi}(f)\propto R_{\mathrm{form}}(f)$ for $f\in\mathcal X$,
where $P_{\phi}$ is the learned formula distribution and
$R_{\mathrm{form}}(f)>0$ is the reward assigned to formula $f$. 
Thus, both formulations retain an individual formula as the generation
object, even when pool performance contributes to its reward. 

We reformulate alpha discovery as learning a conditional generative
model $P_{\theta}(F'\mid F)$ over the space
$\mathcal M_K(\mathcal X)$ of size-$K$ pools of canonically distinct
formulas. The model is trained to
favor candidates with stronger predictive performance and inner
diversity, while using the current pool context to determine how their
formulas are modified:
\begin{equation}
\small
F'\sim P_{\theta}(\cdot\mid F),
\quad
F,F'\in\mathcal M_K(\mathcal X),
\label{eq:pool_aware_generation}
\end{equation}
where $F'$ is generated from $F$ through hierarchical corruption and
reverse generation under the current pool context. By sampling from this learned pool-conditional distribution, we extend
the role of the pool from evaluating completed candidates to guiding
their generation, thereby realizing \emph{pool-aware alpha discovery}.

\highlight{Problem formulation.}
\label{sec:budget}
Following this pool-aware formulation, given a market panel
$\{(X_d,y_d)\}$ divided into training, validation, and test periods and a
formula language $\mathcal X$, the task is to discover an alpha pool
$F=[f_1,\ldots,f_K]\in\mathcal M_K(\mathcal X)$. A ridge combiner~\citep{hoerl1970ridge}, fitted
once on the training split and kept fixed afterwards, jointly fits the
ranked member signals:
\begin{equation}
\small
\widehat w_F=\argmin_w n_{\mathrm{tr}}^{-1}\lVert\widetilde y_{\mathrm{tr}}-\widetilde X_{F,\mathrm{tr}}w\rVert_2^2+\gamma\lVert w\rVert_2^2,\;
\widehat y_{F,S}=\widetilde X_{F,S}\widehat w_F,
\label{eq:ridge}
\end{equation}
where $\widetilde X_{F,S}$ and $\widetilde y_S$ are the stacked tie-aware
daily cross-sectional ranks on split $S$, $n_{\mathrm{tr}}$ is the
number of eligible training observations, and $\gamma$ is the
regularization coefficient. The pool utility $U_S(F)$ is the mean daily
cross-sectional Spearman correlation between $\widehat y_{F,S}$ and
$\widetilde y_S$.

Because executing an unseen formula over the historical panel dominates
the search cost, let $\mathcal A$ contain the distinct canonical formulas
submitted to the evaluator and $\mathcal E\subseteq\mathcal A$ those
with executable, eligible, and nonconstant signals. The objective is formulated as follows:
\begin{equation}
\small
U^{\star}(F^{\star})=\max\{U^{\star}(F):F\in\mathcal M_K(\mathcal E)\},\;
|\mathcal A|\leq B,
\label{eq:budget}
\end{equation}
where $U^{\star}$ is the test utility and $B$ is the budget on distinct
canonical submissions. Training data guide search, validation data
select the final pool, and test data are used only for final evaluation.
The set $\mathcal A$ includes first-time submissions from initialization,
reverse generation, and duplicate replacement, including failed
executions, while cached repeats incur no additional budget.

\section{Methodology}
\label{sec:method}
\begin{figure}[t]
 \centering
 \includegraphics[width=\linewidth]{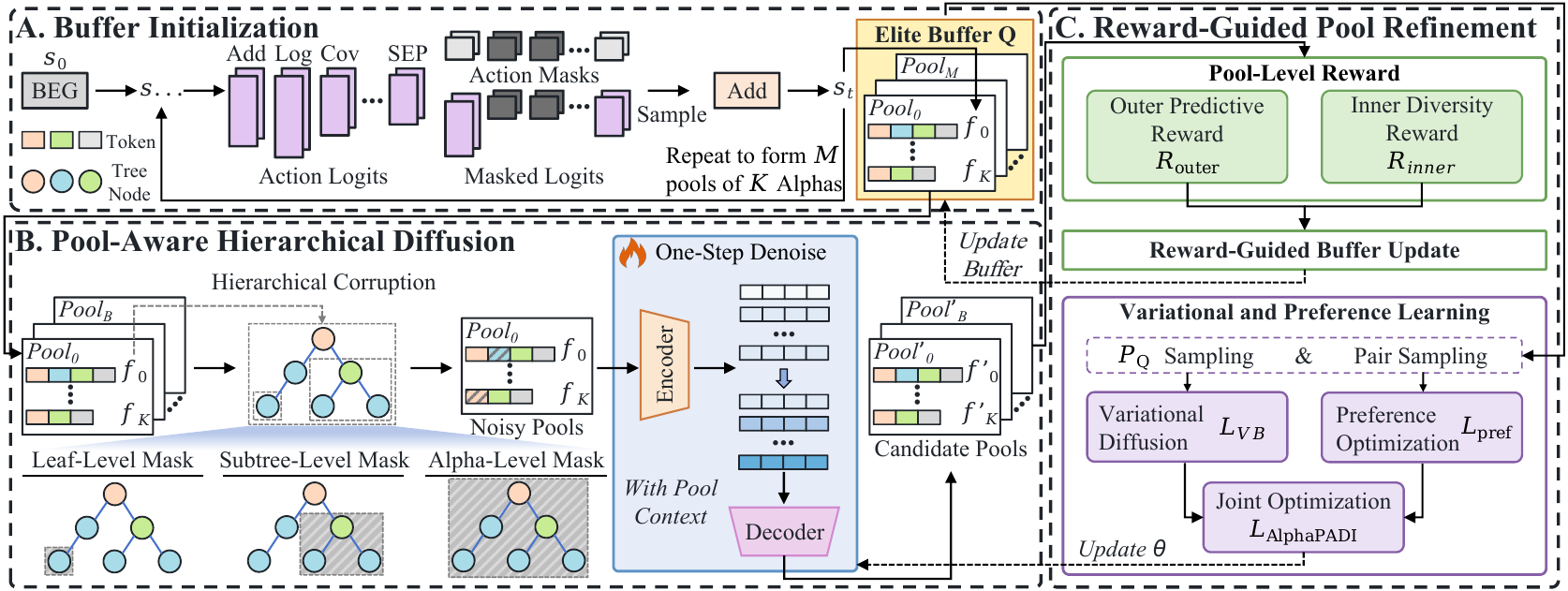}
 \caption{An overview of \shortname{}. A. Buffer initialization: grammar-constrained RPN sampling forms $M$ scored pools of $K$ formulas
 in the elite buffer. B. Pool-aware hierarchical diffusion: buffered pools are corrupted at the leaf, subtree, and alpha levels and reconstructed with the current pool as context. C. Reward-guided pool refinement: Candidates scored by $R_{\mathrm{outer}}$ and $R_{\mathrm{inner}}$ update the buffer, which supplies reconstruction samples and preference pairs
 for $\mathcal L_{\mathrm{AlphaPADI}}$.}
 \label{fig:framework}
\end{figure}

In this section, we propose \shortname{}, which shifts the unit of alpha discovery from individual formulas to pools of complementary formulas. The key insight is that a formula's value depends on its contribution to the joint predictive performance of its pool. We therefore treat the pool as the common unit of generation, evaluation, and learning. As illustrated in Figure~\ref{fig:framework}, \shortname{} comprises \textbf{buffer initialization} (Section~\ref{sec:buffer_initial}), \textbf{pool-aware hierarchical diffusion} (Section~\ref{sec:pool_diffusion}), and \textbf{reward-guided pool refinement} with a joint training objective (Section~\ref{sec:reward_guide}). The notation is summarized in Table~\ref{tab:notation}. We introduce these components below.

\subsection{Buffer Initialization}
\label{sec:buffer_initial}
The elite buffer $\mathcal Q$ stores scored pools that serve as starting
points for refinement and as training data for the generator. We
initialize it through grammar-constrained sampling in Reverse Polish
Notation (RPN)~\citep{yu2023generating} over a vocabulary $\mathcal V$ of
61 feature, operator, window, and constant tokens
(Appendix~\ref{app:features_ops}). The tokens \textsc{beg} and
\textsc{sep} delimit generation, while \textsc{mask} and \textsc{fill}
are reserved for diffusion. Starting from $s_0=\textsc{beg}$ with an
empty stack, features and constants push expressions, unary and binary
operators consume one and two expressions, and each temporal operator
requires a window token. At state $s_t$, the next token is sampled
uniformly from the valid set $\mathcal V_{\mathrm{valid}}(s_t)$, and
\textsc{sep} is valid only when the stack holds one complete expression.
For example, \texttt{(Close 5 TsMean Open Sub SEP)} produces
$\varnothing\!\rightarrow[\texttt{Close}]
\rightarrow[\texttt{Close},\texttt{5}]
\rightarrow[\texttt{TsMean}(\texttt{Close},5)]
\rightarrow[\texttt{TsMean}(\texttt{Close},5),\texttt{Open}]
\rightarrow[\texttt{Sub}(\texttt{TsMean}(\texttt{Close},5),\texttt{Open})]$.
Sequences that exceed $L_{\max}=20$ tokens, fail during execution, or
produce ineligible or zero-variance signals are resampled.

Initial pools are sampled independently of rewards and external
factor libraries. We group $K$ canonically distinct formulas
per pool and construct $M\le C_{\mathcal Q}$ distinct pools:
\begin{equation}
\small
F_0^{(m)}=\left(f_1^{(m)},\ldots,f_K^{(m)}\right),
\mathcal Q_0=\left\{\left(F_0^{(m)},R(F_0^{(m)})\right)\right\}_{m=1}^{M}.
\label{eq:elite_buffer_initialization}
\end{equation}
Each pool is scored by the reward in Eq.~\ref{eq:pool_reward}. All
canonical formulas submitted during initialization, including resampled
ones, enter $\mathcal A$ in Eq.~\ref{eq:budget}, and the scored pools
initialize both $\mathcal Q$ and the validation archive $\mathcal C$.
Before its first proposal, the generator performs $S_0$ updates of
Eq.~\ref{eq:variational_diffusion} on $\mathcal Q_0$.

\subsection{Pool-Aware Hierarchical Diffusion}
\label{sec:pool_diffusion}
Section~\ref{sec:buffer_initial} provides scored pools. Here, we introduce how to revise them in the context of the other pool members: hierarchical corruption masks subtrees at multiple structural scales, and pool-aware reverse generation reconstructs them conditioned on the current pool. Applied to a proposal batch $\mathcal B_{\mathrm c}$ of $N_{\mathrm c}$ buffered pools, this process produces candidates for evaluation.

\highlight{Hierarchical corruption.}
Let $F_0=(f_1,\ldots,f_K)$ be a clean pool and $\mathbf m$ denote
\textsc{mask}. Each formula is parsed into an abstract syntax tree (AST),
in which a node is absorbable if its RPN span forms a complete
expression. For $f=\operatorname{TsMean}(\mathrm{Close},5)-\mathrm{Open}$,
serialized as \texttt{(Close 5 TsMean Open Sub)}, masking the \textit{leaf}
$\mathrm{Close}$, the \textit{subtree} $\operatorname{TsMean}(\mathrm{Close},5)$,
or the \textit{root} yields $\operatorname{TsMean}(\mathbf m,5)-\mathrm{Open}$,
$\mathbf m-\mathrm{Open}$, or $\mathbf m$, corrupting $f$ at the \textit{leaf},
\textit{subtree}, and \textit{alpha} levels, respectively (Component B in Figure~\ref{fig:framework}).
Each absorbable node $u$ receives an independent clock
$U_u\sim\operatorname{Uniform}(0,1)$ shared across levels, and
$\Phi_{\rho}$ masks the highest subtrees with $U_u\leq\rho$, ancestors
taking precedence over descendants. The corrupted pools are formulated as follows:
\begin{equation}
\small
F_t=\Phi_{\rho_t}(F_0),\quad
\rho_t=1-\cos\bigl(\pi t/(2T)\bigr),\quad
t\in\{0,\ldots,T\},\quad
F_T=(\mathbf m,\ldots,\mathbf m).
\label{eq:hier_corruption}
\end{equation}
Let $H_t$ denote the exposed masks in $F_t$ and $Y_u$ the state of
$u\in H_t$ in $F_{t-1}$. The shared clocks induce the posterior
$q(Y_u\mid F_t,F_0)=a_t\delta_{\mathbf m}(Y_u)+w_t\mu_{t-1}^{u,F_0}(Y_u)$
(Proposition~\ref{prop:app_posterior}), where $a_t=\rho_{t-1}/\rho_t$
and $w_t=(\rho_t-\rho_{t-1})/\rho_t$ are the mask-retention and reveal
probabilities, and $\mu_{t-1}^{u,F_0}$ is the level-$(t-1)$ distribution
of the subtree at $u$ with $u$ visible and its descendants corrupted.

\highlight{Pool-aware reverse generation.}
At transition $t\rightarrow t-1$, each position in $H_t$ enters the
reveal set $J_t$ independently with probability $w_t$. Selected positions
keep their RPN order within each formula and are uniformly interleaved
across formulas into $(u_{t,1},\ldots,u_{t,|J_t|})$, and descendant masks
generated at level $t$ are processed at the next level. For the $r$-th
position $u=u_{t,r}$, let $G=G^{(r-1)}=(g_1,\ldots,g_K)$ be the current
pool. A shared Transformer $\psi$ mean-pools each formula, a
\textsc{fill} copy of row $i(u)$ marks the queried position, and a set
encoder aggregates them:
\begin{equation}
\small
\mathbf c_{G,u,t}
=
\operatorname{Mean}\!\left[
\operatorname{Enc}_{\mathrm s}
\left(
\{\psi(g_i)\}_{i=1}^{K},
\psi(g_{i(u)}^{\mathrm{fill}(u)})
\right)
\right]+E_t(t),
\label{eq:pool_context}
\end{equation}
where $E_t(t)$ embeds the diffusion level, and the context $c_{G,u,t}$
summarizes the current pool $G$ for position $u$.
The autoregressive decoder uses this context to sample a
grammar-valid $\hat{Y}_u \sim r_\theta(\cdot \mid G,u,t)$.
We update $G^{(r)}=\operatorname{Splice}(G,u,\hat{Y}_u)$
before filling the next mask.
This lets each reconstruction use earlier edits and capture
dependencies across formulas (Proposition~\ref{prop:app_reverse}).
Because $w_1=1$, every remaining mask is reconstructed at $t=1$, yielding
complete and syntactically valid pools.

\highlight{Candidate generation.}
Each pool in $\mathcal B_{\mathrm c}$ is corrupted to
$t_0=\max\{1,\operatorname{round}(\alpha T)\}$, where $\alpha\in(0,1]$
controls the edit strength, and reconstructed from $t_0$ to $0$. 
After canonicalization, repeated members are replaced by grammar-valid random
formulas so that each pool contains $K$ distinct formulas. Members not
yet in $\mathcal A$ are submitted to the evaluator, and a candidate whose
unseen members would exceed the remaining budget is discarded. Eligible
candidates (Section~\ref{sec:reward_guide}) are scored and added,
together with the initial pools, to the archive $\mathcal C$. Search
stops when the budget is exhausted or $N_{\mathrm{stall}}$ consecutive
rounds add no formula to $\mathcal A$, and the final pool is
$\widehat F=\argmax_{F\in\mathcal C}U_{\mathrm{val}}(F)$.

\subsection{Reward-Guided Pool Refinement and Training Objective}
\label{sec:reward_guide}
Section~\ref{sec:pool_diffusion} produces candidate pools. Here, we score them with a pool-level reward and use this reward in two ways: to refine the elite buffer and to train the reverse model through reconstruction and preference objectives.

\highlight{Pool-level reward.}
For each date $d$, member $f_i$ produces a signal $\mathbf a_d(f_i)$,
whose tie-aware cross-sectional ranks $\widetilde{\mathbf a}_{i,d}$ are
computed over assets with finite values for all members and the target.
The outer reward is the training performance of the ridge combiner in
Eq.~\ref{eq:ridge}, whose penalty is selected within the training period.
Ridge controls the combination weights but does not prevent redundant
formulas from occupying the $K$ pool positions. We therefore measure
redundancy by
$c_{ij}=|\mathcal T_{ij}|^{-1}\sum_{d\in\mathcal T_{ij}}
|\operatorname{corr}(\widetilde{\mathbf a}_{i,d},\widetilde{\mathbf a}_{j,d})|$
and let $d_i=\max_{j\ne i}c_{ij}$. The reward components, the pool
reward, and the buffer sampling distribution are formulated as follows:
\begingroup
\small 
\begin{gather}
R_{\mathrm{outer}}(F)
=\frac{1}{|\mathcal T_F|}\sum_{d\in\mathcal T_F}
\operatorname{Spearman}\bigl(\widehat{\mathbf y}_d(F),\mathbf y_d\bigr),
\quad
R_{\mathrm{inner}}(F)
=-\frac{1}{K}\sum_{i=1}^{K}[d_i-\delta]_+^2,
\label{eq:pool_reward_components}\\
R(F)=R_{\mathrm{outer}}(F)
+w_{\mathrm{inner}}R_{\mathrm{inner}}(F),
\label{eq:pool_reward}\\
P_{\mathcal Q}(F_i)
=\frac{\exp(z_i/T_{\mathrm r})}
{\sum_j\exp(z_j/T_{\mathrm r})},
\quad
z_i=\frac{R_i-\overline R}
{\operatorname{sd}_{\mathcal Q}(R)+\epsilon},
\label{eq:buffer_sampling}
\end{gather}
\endgroup
where $\mathcal T_{ij}$ contains the dates with jointly finite,
nonconstant ranked signals ($c_{ij}=0$ if $\mathcal T_{ij}=\varnothing$),
$\mathcal T_F$ contains the dates with nonconstant pool predictions and
targets, $[x]_+=\max(x,0)$, $\delta\in[0,1]$ is the tolerated
similarity, $w_{\mathrm{inner}}>0$ weights diversity, $\overline R$ and $\operatorname{sd}_{\mathcal Q}(R)$ are the mean and standard deviation of the buffered rewards, $\epsilon>0$
stabilizes reward standardization, and $T_{\mathrm r}$ controls the
sampling concentration. A pool is excluded if it has no eligible date or
a member without a finite, nonconstant training signal. The elite buffer
retains the highest-reward canonical pools up to capacity $C_{\mathcal Q}$
and samples them from $P_{\mathcal Q}$.

\highlight{Reconstruction objective.}
To keep the reverse model able to reconstruct high-reward pools, we
minimize a variational bound on the negative log-likelihood of buffered
pools. Since the reveal probabilities and the reconstruction order are
shared by the forward posterior and the reverse model, they cancel, and
the bound reduces, up to a constant, to $T$ times the expected negative
log-likelihood of the revealed targets under a uniformly sampled level
(Proposition~\ref{prop:app_elbo}). We estimate this bound on a minibatch
$\mathcal B_{\mathcal Q}=\{F_0^{(b)}\}_{b=1}^{N}$ sampled from
$P_{\mathcal Q}$. For each pool, we draw a level
$t_b\sim\operatorname{Uniform}(\{1,\ldots,T\})$ and its reveal set $J_b$,
and teacher-force the targets $Y_{b,r}$, the states of the revealed
positions $u_{b,r}$ at level $t_b-1$, into the pool state $G_b^{(r)}$,
starting from $G_b^{(0)}=F_{t_b}^{(b)}$. Normalizing each pool by the
number and length of its targets gives:
\begin{equation}
\small
\widehat{\mathcal L}_{\mathrm{VB}}(\mathcal B_{\mathcal Q},\theta)
=\frac{T}{N}\sum_{b=1}^{N}\ell_\theta^{(b)},
\quad
\ell_\theta^{(b)}
=-\frac{1}{n_b}\sum_{r=1}^{|J_b|}
\frac{\log r_\theta\bigl(Y_{b,r}\mid G_b^{(r-1)},u_{b,r},t_b\bigr)}{L_{b,r}},
\label{eq:variational_diffusion}
\end{equation}
where $L_{b,r}$ is the length of $Y_{b,r}$ including \textsc{sep} and
$n_b=\max\{1,|J_b|\}$. The normalization makes $-\ell_\theta^{(b)}$
comparable across pools with different numbers and lengths of targets,
which the preference objective below relies on.

\highlight{Pool-level preference objective.}
We convert pool rewards into preferences by drawing two distinct pools
from $P_{\mathcal Q}$, denoting the higher-reward pool by
$F_{\mathrm w}$ and the lower-reward pool by $F_{\mathrm l}$, and
retaining the pair when
$R(F_{\mathrm w})-R(F_{\mathrm l})\geq
m_{\mathrm{dpo}}\operatorname{sd}_{\mathcal Q}(R)$, where $m_{\mathrm{dpo}}\ge0$
sets the reward margin in units of the buffer reward standard deviation.
Ties and pairs with empty reconstruction sets are omitted, and
$\mathcal P_{\mathcal Q}$ denotes the distribution of accepted pairs. Both pools of a pair share
the diffusion level and the pseudorandom stream for corruption, reveal
selection, and ordering, while each consumes the stream according to its
own AST and thus keeps its marginal corruption process. The objective is formulated as follows:
\begin{equation}
\small
\mathcal L_{\mathrm{pref}}(\theta)
=
-\mathbb E_{\mathcal P_{\mathcal Q}}
\log\sigma\!\left(
\beta\left[
\left(\ell_\theta^{\mathrm l}-\ell_\theta^{\mathrm w}\right)
-
\left(
\ell_{\theta_{\mathrm{ref}}}^{\mathrm l}
-\ell_{\theta_{\mathrm{ref}}}^{\mathrm w}
\right)
\right]
\right),
\label{eq:pool_preference}
\end{equation}
where $\ell_\vartheta^{a}$ is the single-pool loss of
Eq.~\ref{eq:variational_diffusion} for $a\in\{\mathrm w,\mathrm l\}$ and
$\vartheta\in\{\theta,\theta_{\mathrm{ref}}\}$, evaluated under the same
corruption variables, $\theta_{\mathrm{ref}}$ is a frozen copy of the
reverse model after the $S_0$ initialization updates, $\sigma$ is the
sigmoid function, and $\beta>0$ controls the preference strength and
absorbs the factor $T$. Minimizing Eq.~\ref{eq:pool_preference}
increases the reference-relative reconstruction score of the
higher-reward pool, and Proposition~\ref{prop:app_pref} derives it from
KL-regularized pool-level reward maximization.

\highlight{Joint training objective.}
We combine the preference objective with reconstruction on an
independently sampled minibatch:
\begin{equation}
\small
\mathcal L_{\mathrm{AlphaPADI}}(\theta)
=
\mathcal L_{\mathrm{pref}}(\theta)
+
\lambda_{\mathrm{VB}}
\widehat{\mathcal L}_{\mathrm{VB}}
(\mathcal B_{\mathcal Q},\theta),
\label{eq:joint_objective}
\end{equation}
where $\lambda_{\mathrm{VB}}>0$ weights reconstruction, anchoring
the generator to the buffered pools while the preference term contrasts
them. If no pair is accepted, the update uses reconstruction only.
Algorithm~\ref{alg:pool_refinement} in Appendix~\ref{app:algorithm}
summarizes the complete procedure, and Appendix~\ref{app:complexity}
analyzes its complexity.

\subsection{Discussion of Alpha Combination}
\label{sec:discussion}
We evaluate every pool, including those returned by the baselines, with
the ridge combiner in Eq.~\ref{eq:ridge}, fitted once on the training
period and frozen during validation and testing. This aligns evaluation with the predictive component of \shortname{}'s pool-level reward, holds the combination procedure fixed across pools, and uses the same backtesting protocol for all methods (Appendix~\ref{app:implementation_details}).

\section{Experiments}
\label{sec:exp}

\subsection{Experimental Setup}
\label{sec:experimental_setup}
\highlight{Datasets.}
We evaluate CSI300, CSI500, and S\&P500, representing large-cap Chinese,
mid-cap Chinese, and large-cap U.S. equities, respectively. Following
formulaic alpha studies~\citep{shi2025alphaforge,chen2026alphasage},
we use Qlib daily market data~\citep{yang2020qlib} and
predict the 20-trading-day forward return. For consistent
cross-market evaluation, three universes share a chronological split: 2016-01-01--2020-12-31 for training,
2021-01-01--2022-12-31 for validation, and
2023-01-01--2025-12-31 for testing. Dates whose return horizon crosses a
split boundary are excluded.

\highlight{Evaluation metrics.}
We evaluate predictive performance using the information coefficient (IC), rank information coefficient (RIC), and their information ratios (ICIR and RICIR), and portfolio performance using annualized return (AR), Sharpe ratio (SR), and maximum drawdown (MDD). Lower MDD is better, and higher values are better for the other metrics. 
Metric definitions and the shared backtesting protocol are provided in Appendix~\ref{app:evaluation_protocol}.

\highlight{Baselines.}
We compare \shortname{} against nine baselines spanning the main
paradigms of alpha mining, including the predefined factor library
Alpha158~\citep{yang2020qlib}, the supervised predictors MLP~\citep{murtagh1991multilayer} and
XGBoost~\citep{chen2016xgboost}, genetic programming (GP)~\citep{chen2021empirical}, the
reinforcement learning methods AlphaGen~\citep{yu2023generating} and
AlphaQCM~\citep{zhu2025alphaqcm}, the generative-predictive framework
AlphaForge~\citep{shi2025alphaforge}, the LLM-based miner
AlphaAgent~\citep{tang2025alphaagent}, and the GFlowNet-based
AlphaSAGE~\citep{chen2026alphasage}. Appendix~\ref{app:baselines} details all baselines and their settings.

\highlight{Implementation details.}
All experiments run on 8 NVIDIA L40S GPUs. Non-LLM automated
alpha mining methods share a budget of $40{,}000$ formula evaluations,
converted from native units, e.g., $2\times10^5$ steps (AlphaGen)
or $4\times10^4$ episodes (AlphaSAGE) (Appendix~\ref{app:baselines} specifies each method's stopping rule). We tune the learning rate over
$\{2\times10^{-5},3\times10^{-4},1\times10^{-3}\}$,
$w_{\mathrm{inner}}$ over $\{1.0,2.0,5.0,10.0\}$, and
$\lambda_{\mathrm{VB}}$ over $\{0.1,0.3,0.5,0.7\}$ on the validation
split, and select the pool size $K\in\{20,50\}$ likewise for
non-LLM automated alpha mining baselines and \shortname{}. We report
the mean performance over five runs and assess significance against the
strongest baseline at $p<0.05$.
Appendix~\ref{app:hyper_settings} lists the architecture and
hyperparameters, and Appendix~\ref{app:hyperparameter_analyses} analyzes sensitivity to $w_{\mathrm{inner}}$ and $\lambda_{\mathrm{VB}}$.

\subsection{Overall Performance}

\begin{table*}[t]
\centering
\caption{
Predictive and portfolio performance across three equity universes.
RIC and RICIR denote RankIC and RankICIR.
Higher values are better except for MDD.
Best results are \textbf{bolded}, and second-best results are
\underline{underlined}.
}
\label{tab:main_results}
\footnotesize
\rmfamily
\setlength{\tabcolsep}{2.8pt}
\renewcommand{\arraystretch}{1.0}

\scalebox{0.885}{%
\begin{tabularx}{\textwidth}{
>{\centering\arraybackslash}m{1.55cm}
>{\centering\arraybackslash}m{0.75cm}|
*{4}{>{\centering\arraybackslash}X}|
*{3}{>{\centering\arraybackslash}X}
}
\hline
\multirow[c]{2}{1.55cm}{\centering\textbf{Method}}
& \multirow[c]{2}{0.75cm}{\centering\textbf{Year}}
& \multicolumn{4}{c|}{\textbf{Predictive Performance}}
& \multicolumn{3}{c}{\textbf{Portfolio Performance}} \\
\cline{3-6}\cline{7-9}
& &
\textbf{IC$\uparrow$}
& \textbf{ICIR$\uparrow$}
& \textbf{RIC$\uparrow$}
& \textbf{RICIR$\uparrow$}
& \textbf{AR$\uparrow$}
& \textbf{SR$\uparrow$}
& \textbf{MDD$\downarrow$} \\
\hline

\multicolumn{9}{c}{\textbf{CSI300}} \\
\hline
Alpha158 & 2020
& 0.0266 & 0.1631 & 0.0250 & 0.1389
& 0.0751 & \underline{2.1532} & 0.2274 \\
XGBoost & 2016
& 0.0192 & 0.1024 & 0.0273 & 0.1386
& 0.0591 & 1.8645 & 0.2046 \\
MLP & --
& 0.0254 & 0.1458 & 0.0362 & 0.2047
& 0.0723 & 1.8562 & 0.2755 \\
GP & 2021
& 0.0229 & 0.2142 & 0.0326 & 0.2788
& 0.0721 & 1.6105 & 0.3247 \\
AlphaGen & 2023
& 0.0290 & 0.1771 & 0.0380 & 0.2065
& 0.0657 & 2.1058 & \underline{0.1829} \\
AlphaQCM & 2025
& \underline{0.0352} & \underline{0.2291}
& \underline{0.0544} & 0.3371
& 0.0733 & 1.4941 & 0.3787 \\
AlphaForge & 2025
& 0.0276 & 0.1391 & 0.0516 & 0.2501
& \underline{0.0940} & 1.5646 & 0.3669 \\
AlphaAgent & 2025
& 0.0251 & 0.1617 & 0.0510 & \underline{0.3388}
& 0.0713 & 1.3969 & 0.3954 \\
AlphaSAGE & 2026
& 0.0228 & 0.1791 & 0.0351 & 0.2656
& 0.0752 & 1.5401 & 0.3939 \\
\hline
\rowcolor{gray!22}
\textbf{AlphaPADI} & 2026
& \textbf{0.0439} & \textbf{0.3024}
& \textbf{0.0659} & \textbf{0.4147}
& \textbf{0.1058} & \textbf{2.8619} & \textbf{0.1412} \\
\hline

\multicolumn{9}{c}{\textbf{CSI500}} \\
\hline
Alpha158 & 2020
& 0.0108 & 0.0952 & 0.0163 & 0.1327
& 0.1188 & 2.5419 & 0.2565 \\
XGBoost & 2016
& 0.0082 & 0.0564 & 0.0251 & 0.1810
& 0.0957 & 2.3123 & 0.2246 \\
MLP & --
& 0.0086 & 0.0648 & 0.0277 & 0.2126
& 0.0898 & 1.9643 & 0.3198 \\
GP & 2021
& 0.0111 & 0.0936 & 0.0251 & 0.1892
& 0.1132 & 2.6883 & 0.2366 \\
AlphaGen & 2023
& 0.0203 & 0.1571 & 0.0379 & 0.2641
& 0.1233 & 2.9673 & 0.1830 \\
AlphaQCM & 2025
& 0.0167 & 0.1232 & 0.0453 & 0.2996
& 0.1154 & 2.5779 & 0.2190 \\
AlphaForge & 2025
& 0.0184 & \underline{0.2583} & 0.0269
& \underline{0.3484}
& 0.1175 & 2.2819 & 0.3198 \\
AlphaAgent & 2025
& 0.0092 & 0.0693 & 0.0278 & 0.2104
& 0.1158 & 2.3852 & 0.2437 \\
AlphaSAGE & 2026
& \underline{0.0235} & 0.1413 & \underline{0.0601}
& 0.3477
& \underline{0.1372} & \underline{3.1978}
& \underline{0.1687} \\
\hline
\rowcolor{gray!22}
\textbf{AlphaPADI} & 2026
& \textbf{0.0444} & \textbf{0.2855}
& \textbf{0.0763} & \textbf{0.4578}
& \textbf{0.1435} & \textbf{3.5124} & \textbf{0.1350} \\
\hline

\multicolumn{9}{c}{\textbf{S\&P500}} \\
\hline
Alpha158 & 2020
& 0.0221 & 0.1482 & 0.0194 & 0.1291
& 0.1572 & 3.8617 & 0.1204 \\
XGBoost & 2016
& 0.0120 & 0.1079 & 0.0073 & 0.0688
& 0.1498 & 4.0652 & 0.1096 \\
MLP & --
& 0.0175 & 0.1893 & 0.0143 & 0.1588
& 0.1517 & 4.5017 & 0.1041 \\
GP & 2021
& 0.0142 & 0.3193 & 0.0113 & 0.2700
& 0.1347 & 4.3319 & 0.0958 \\
AlphaGen & 2023
& 0.0174 & 0.1298 & 0.0103 & 0.0786
& 0.1571 & 3.9775 & 0.1102 \\
AlphaQCM & 2025
& 0.0270 & 0.2620 & 0.0145 & 0.1458
& \underline{0.1704} & 4.7981 & 0.1064 \\
AlphaForge & 2025
& 0.0315 & \underline{0.4414}
& \underline{0.0220} & \underline{0.3122}
& 0.1624 & \underline{5.3919} & \underline{0.0952} \\
AlphaAgent & 2025
& 0.0193 & 0.2153 & 0.0138 & 0.1462
& 0.1428 & 4.6385 & 0.0987 \\
AlphaSAGE & 2026
& \underline{0.0319} & 0.3099
& \underline{0.0220} & 0.2151
& 0.1679 & 4.6631 & 0.0997 \\
\hline
\rowcolor{gray!22}
\textbf{AlphaPADI} & 2026
& \textbf{0.0515} & \textbf{0.5020}
& \textbf{0.0452} & \textbf{0.4122}
& \textbf{0.2061} & \textbf{5.8423} & \textbf{0.0860} \\
\hline
\end{tabularx}%
}
\end{table*}
Table~\ref{tab:main_results} reports predictive and portfolio performance on the three stock universes, and Figure~\ref{fig:cr_plot_csi300} shows the CSI300 cumulative returns over the test period. 
Backtests on CSI500 and S\&P500, initialization comparisons, additional
portfolio metrics, and search efficiency under the shared budget are
provided in Appendices~\ref{app:backtesting},~\ref{app:cold_start},
~\ref{app:other_metrics}, and~\ref{app:efficiency}.

\noindent\begin{minipage}{\textwidth}
\setlength{\intextsep}{4pt}
\begin{wrapfigure}{r}{0.50\columnwidth}
    \centering
    \includegraphics[width=\linewidth]{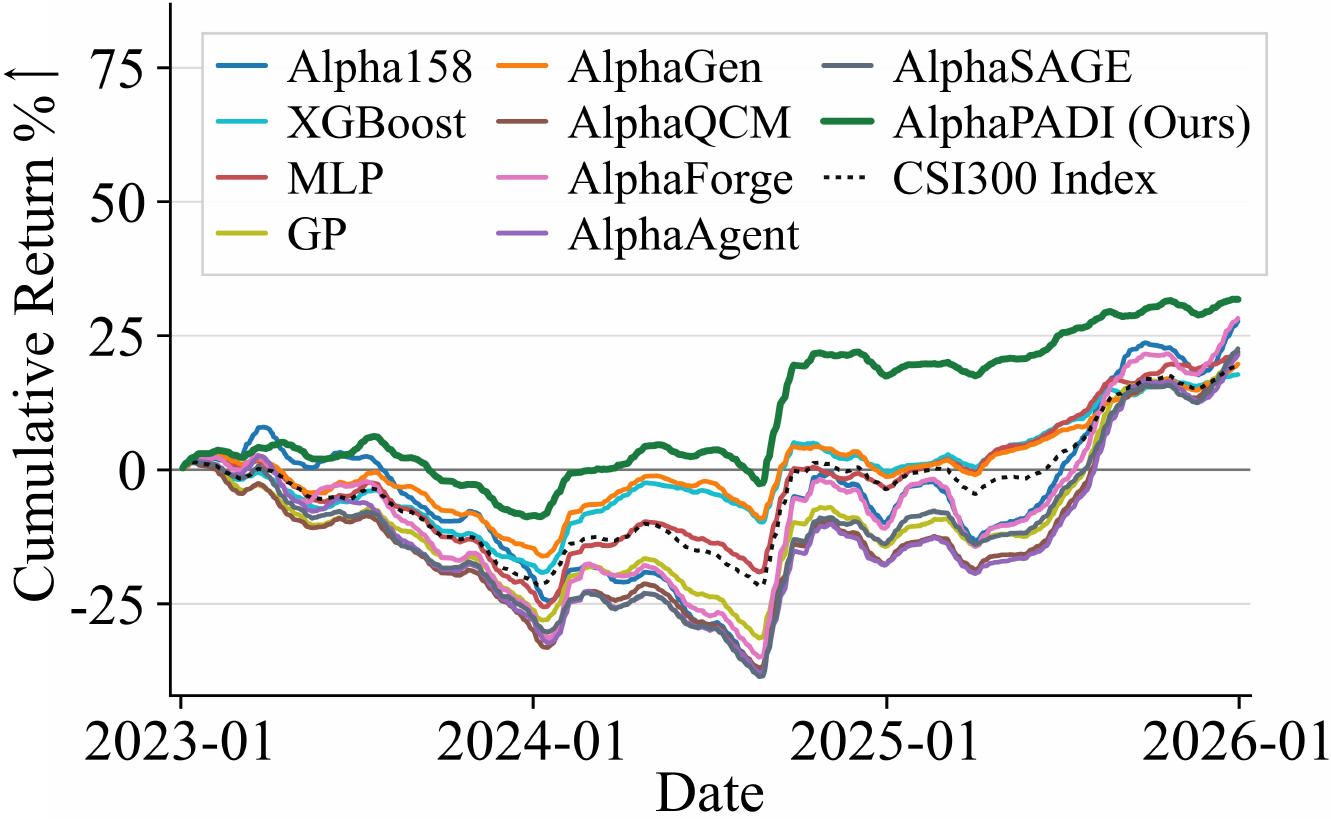}
    \vspace{-4mm}
    \caption{Cumulative returns on CSI300 (2023--2025): \shortname{}, baseline methods, and the CSI300 index.}
    \label{fig:cr_plot_csi300}
\end{wrapfigure}
\textbf{Obs 1: \shortname{} achieves the best performance on all seven metrics in all three universes.} Its RIC reaches $0.0659$, $0.0763$, and $0.0452$ on CSI300, CSI500, and S\&P500, improving over the strongest baseline by $21.1\%$, $27.0\%$, and $105.5\%$, respectively. These gains are also reflected in portfolio metrics: AlphaPADI achieves the highest AR and SR across all three markets, while reducing MDD relative to the best baseline by 22.8\%, 20.0\%, and 9.7\% on CSI300, CSI500, and S\&P500, respectively.  Under a shared protocol, searching over complete pools yields more robust combinations than generating formulas individually. \\
\textbf{Obs 2: No baseline is consistently the strongest.}
The second-best entries in Table~\ref{tab:main_results} span six baselines, and the leading baseline varies by both universe and metric: AlphaQCM ranks second in IC on CSI300, AlphaSAGE on CSI500, and AlphaForge in ICIR and SR on S\&P500. \\
\textbf{Obs 3: \shortname{} leads from the 2024 drawdown through recovery.}
On CSI300, \shortname{} has a shallower trough during the decline before September 2024 and finishes above all baselines after the subsequent recovery (Figure~\ref{fig:cr_plot_csi300}), consistent with its lower MDD and higher SR.

\end{minipage}

\newcommand{\cmark}{\ding{51}}
\newcommand{\xmark}{\ding{55}}

\newcolumntype{C}[1]{%
    >{\centering\arraybackslash}m{#1}%
}
\newcolumntype{Y}{%
    >{\centering\arraybackslash}X%
}

\begin{table*}[htb]
\centering

\caption{
Ablation study of \shortname{} on CSI300. The base model is Diffusion w/o pool context.
$\mathcal L_{\mathrm{VB}}$ denotes pool-aware hierarchical diffusion in Eq.~\ref{eq:variational_diffusion}.
$R_{\mathrm{outer}}$ and $R_{\mathrm{inner}}$ denote the
outer predictive reward and inner diversity reward in Eq.~\ref{eq:pool_reward_components},
respectively.
$\mathcal L_{\mathrm{pref}}$ denotes pool-level preference optimization in Eq.~\ref{eq:pool_preference}.
Best results are \textbf{bolded}, and second-best results are
\underline{underlined}.
}

\label{tab:ablation}

\footnotesize
\rmfamily
\setlength{\tabcolsep}{2.8pt}
\renewcommand{\arraystretch}{1.0}

\begin{tabularx}{0.9\textwidth}{
    *{11}{Y}
}

\toprule[1.1pt]

\multicolumn{4}{c}{\textbf{Included Components}}
&
\multicolumn{4}{c}{\textbf{Predictive Performance}}
&
\multicolumn{3}{c}{\textbf{Portfolio Performance}}
\\[-1pt]

\cmidrule(r{0.35em}){1-4}
\cmidrule(l{0.35em}r{0.35em}){5-8}
\cmidrule(l{0.35em}){9-11}

\textbf{$\mathcal L_{\mathrm{VB}}$}
&
$\mathbf{R}_{\mathrm{outer}}$
&
$\mathbf{R}_{\mathrm{inner}}$
&
\textbf{$\mathcal L_{\mathrm{pref}}$}
&
\textbf{IC$\uparrow$}
&
\textbf{ICIR$\uparrow$}
&
\textbf{RIC$\uparrow$}
&
\textbf{RICIR$\uparrow$}
&
\textbf{AR$\uparrow$}
&
\textbf{SR$\uparrow$}
&
\textbf{MDD$\downarrow$}
\\

\midrule

\xmark
&
\xmark
&
\xmark
&
\xmark
&
0.0215
&
0.1434
&
0.0334
&
0.2060
&
0.0780
&
1.9722
&
0.2708
\\

\cmark
&
\xmark
&
\xmark
&
\xmark
&
0.0276
&
0.2263
&
0.0386
&
0.3019
&
0.0731
&
1.9039
&
0.2249
\\

\cmark
&
\cmark
&
\xmark
&
\xmark
&
0.0354
&
\underline{0.2684}
&
0.0463
&
\underline{0.3255}
&
0.0881
&
2.3414
&
0.2333
\\

\cmark
&
\cmark
&
\cmark
&
\xmark
&
\underline{0.0396}
&
0.2510
&
\underline{0.0531}
&
0.3130
&
\underline{0.0960}
&
\underline{2.6810}
&
\underline{0.1915}
\\

\specialrule{\lightrulewidth}{0pt}{0pt}

\rowcolor{gray!10}
\cmark
&
\cmark
&
\cmark
&
\cmark
&
\textbf{0.0439}
&
\textbf{0.3024}
&
\textbf{0.0659}
&
\textbf{0.4147}
&
\textbf{0.1058}
&
\textbf{2.8619}
&
\textbf{0.1412}
\\

\specialrule{1.1pt}{0pt}{0pt}

\end{tabularx}
\end{table*}

\begin{figure}[!t]
    \centering
    \includegraphics[width=0.92\textwidth]{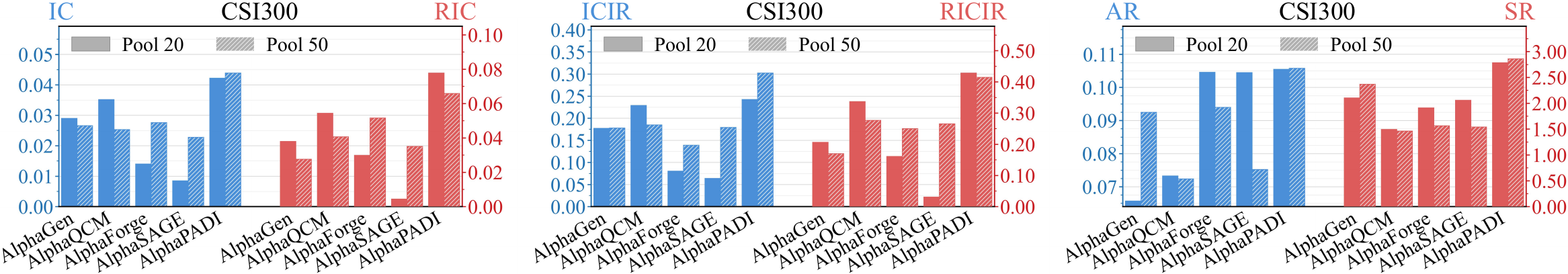}
\caption{Performance of five alpha-mining methods with pool sizes $K=20$ (solid) and $K=50$ (hatched) on CSI300. Blue and red bars use the corresponding left and right metric axes.}
    \label{fig:pool_size_csi300}
    \vspace{-2mm}
\end{figure}

\subsection{Pool-Aware and Component Analysis}
\label{sec:ablation}

Table~\ref{tab:ablation} adds components incrementally on CSI300 under
identical budget and evaluation. The base model reconstructs masked subtrees without pool context and, lacking a pool reward, buffers the latest candidates and samples them uniformly. The second row adds pool context
$\mathbf c_{G,u,t}$. Rows three and four rank and sample buffered pools
by Eq.~\ref{eq:pool_reward} with $w_{\mathrm{inner}}=0$
($R_{\mathrm{outer}}$ alone) and $w_{\mathrm{inner}}>0$. The last row
adds $\mathcal L_{\mathrm{pref}}$, which needs reward-ordered pairs and
therefore follows both rewards.

\textbf{Obs 4: Every component improves IC and RIC, and preference optimization gives the largest gain in RIC and RICIR.}
IC and RIC increase monotonically along the sequence, from $0.0215$ and $0.0334$ to $0.0439$ and $0.0659$. Adding $\mathcal L_{\mathrm{pref}}$ raises RIC from $0.0531$ to $0.0659$ and RICIR from $0.3130$ to $0.4147$, the largest single increments of both metrics, which suggests that reward-ordered pool comparisons provide a learning signal beyond reward-guided buffer selection.\\
\textbf{Obs 5: The diversity reward mainly improves portfolio risk.}
Adding $R_{\mathrm{inner}}$ reduces MDD from $0.2333$ to $0.1915$ and raises SR from $2.3414$ to $2.6810$, while ICIR decreases slightly from $0.2684$ to $0.2510$. Penalizing redundant members thus trades a small loss in predictive stability for lower drawdown. Pool context alone improves ICIR and RICIR but not AR or SR, indicating that pool-level rewards are needed to turn pool-aware reconstruction into trading gains.

\subsection{Effect of Pool Size}
\label{sec:pool_size}
Figure~\ref{fig:pool_size_csi300} compares pools of $K=20/50$ on CSI300 under the same budget of $40{,}000$ formula evaluations for the five methods that construct alpha pools, whereas Table~\ref{tab:main_results} reports each of them at its validation-selected size. The CSI500 and S\&P500 results are provided in Appendix~\ref{app:pool_size}.

\textbf{Obs 6: \shortname{} leads at both pool sizes.}
At both $K=20$ and $50$, \shortname{} attains the highest IC, RIC, ICIR, RICIR, and SR, and its AR is the highest or on par with the strongest baseline.\\
\textbf{Obs 7: A larger pool does not uniformly help.}
From $K=20$ to $50$, predictive and portfolio metrics often move in opposite directions: AlphaForge and AlphaSAGE gain IC and RIC but lose AR and SR, AlphaGen shows the reverse pattern, and AlphaQCM declines on most metrics. For \shortname{}, IC, ICIR, and SR increase, while RIC and RICIR decrease. Pool capacity should therefore be evaluated through the combined predictor rather than assumed to help, which motivates optimizing the pool-level reward directly.

\section{Conclusion}
\label{sec:conclusion}
We presented \shortname{}, a framework that treats the alpha pool as the common unit of generation, evaluation, and learning. Pool-aware hierarchical discrete diffusion modifies formula structures using the current pool context, while predictive and redundancy rewards guide buffer refinement and preference learning. Experiments on three stock universes show higher reported predictive and portfolio performance than the evaluated baselines. Cumulative ablations characterize the incremental effects of the components, and additional analyses examine pool size, initialization, and hyperparameter sensitivity. The findings support pool-level search under the evaluated settings.

\par\endgroup

\subsection*{AI Use Statement}
In this work, we used generative AI tools to obtain feedback on the methodology and experimental evaluation, check the consistency of descriptions of reported results with the tables, and suggest potentially relevant references. We also used these tools to improve the organization and clarity of selected passages, refine English wording, and assist with LaTeX formatting. The authors reviewed all AI-assisted suggestions, decided which to adopt, and finalized the manuscript. The authors take full responsibility for the paper's text, claims, analyses, and results.

\subsection*{Ethics statement}

AlphaPADI studies formulaic alpha discovery using historical equity-market data from China and the United States. The experiments consist of offline signal evaluation and backtesting; they involve no human participants, personal data, brokerage accounts, or live market orders. The reported results depend on the studied markets, periods, and backtesting assumptions. They should not be interpreted as investment recommendations or guarantees of future returns. Deploying a discovered strategy would require further assessment of execution costs and financial risk.

\subsection*{Reproducibility statement}

Section 4 specifies the pool generator, reward, and training objectives, while Appendix A.1 lists the formula vocabulary and Appendix B.2 gives the complete algorithm. Section 5 describes the datasets, chronological splits, baselines, and metrics. Appendices C.1.1--C.1.4 document the model settings, ridge combiner, backtesting protocol, method-specific formula-evaluation budgets, and hyperparameters. We report mean results over five runs, and the paper links the anonymized code. Because initialization, diffusion sampling, and training are stochastic, repeated runs may select different alpha pools.

\bibliography{iclr2027/iclr2027_conference}

@article{yang2020qlib,
  title={Qlib: An ai-oriented quantitative investment platform},
  author={Yang, Xiao and Liu, Weiqing and Zhou, Dong and Bian, Jiang and Liu, Tie-Yan},
  journal={arXiv preprint arXiv:2009.11189},
  year={2020}
}

@article{gu2020empirical,
  title={Empirical asset pricing via machine learning},
  author={Gu, Shihao and Kelly, Bryan and Xiu, Dacheng},
  journal={The review of financial studies},
  volume={33},
  number={5},
  pages={2223--2273},
  year={2020},
  publisher={Oxford University Press}
}

@inproceedings{chen2026alphasage,
  title={Alphasage: Structure-aware alpha mining via gflownets for robust exploration},
  author={Chen, Binqi and Ding, Hongjun and Shen, Ning and Guo, Taian and Huang, Jinsheng and Liu, Luchen and Zhang, Ming},
  booktitle={International Conference on Learning Representations},
  volume={2026},
  pages={2955--2982},
  year={2026}
}

@inproceedings{yu2023generating,
  title={Generating synergistic formulaic alpha collections via reinforcement learning},
  author={Yu, Shuo and Xue, Hongyan and Ao, Xiang and Pan, Feiyang and He, Jia and Tu, Dandan and He, Qing},
  booktitle={Proceedings of the 29th ACM SIGKDD conference on knowledge discovery and data mining},
  pages={5476--5486},
  year={2023}
}

@inproceedings{zhu2025alphaqcm,
  title={Alphaqcm: Alpha discovery in finance with distributional reinforcement learning},
  author={Zhu, Zhoufan and Zhu, Ke},
  booktitle={Forty-second International Conference on Machine Learning},
  year={2025}
}

@inproceedings{shi2025alphaforge,
  title={Alphaforge: A framework to mine and dynamically combine formulaic alpha factors},
  author={Shi, Hao and Song, Weili and Zhang, Xinting and Shi, Jiahe and Luo, Cuicui and Ao, Xiang and Arian, Hamid and Seco, Luis Angel},
  booktitle={Proceedings of the AAAI conference on artificial intelligence},
  volume={39},
  number={12},
  pages={12524--12532},
  year={2025}
}

@inproceedings{tang2025alphaagent,
  title={Alphaagent: Llm-driven alpha mining with regularized exploration to counteract alpha decay},
  author={Tang, Ziyi and Chen, Zechuan and Yang, Jiarui and Mai, Jiayao and Zheng, Yongsen and Wang, Keze and Chen, Jinrui and Lin, Liang},
  booktitle={Proceedings of the 31st ACM SIGKDD Conference on Knowledge Discovery and Data Mining V. 2},
  pages={2813--2822},
  year={2025}
}

@article{zhang2020autoalpha,
  title={Autoalpha: an efficient hierarchical evolutionary algorithm for mining alpha factors in quantitative investment},
  author={Zhang, Tianping and Li, Yuanqi and Jin, Yifei and Li, Jian},
  journal={arXiv preprint arXiv:2002.08245},
  year={2020}
}

@article{austin2021structured,
  title={Structured denoising diffusion models in discrete state-spaces},
  author={Austin, Jacob and Johnson, Daniel D and Ho, Jonathan and Tarlow, Daniel and Van Den Berg, Rianne},
  journal={Advances in neural information processing systems},
  volume={34},
  pages={17981--17993},
  year={2021}
}

@article{sahoo2024simple,
  title={Simple and effective masked diffusion language models},
  author={Sahoo, Subham S and Arriola, Marianne and Schiff, Yair and Gokaslan, Aaron and Marroquin, Edgar and Chiu, Justin T and Rush, Alexander and Kuleshov, Volodymyr},
  journal={Advances in Neural Information Processing Systems},
  volume={37},
  pages={130136--130184},
  year={2024}
}

@article{vignac2022digress,
  title={Digress: Discrete denoising diffusion for graph generation},
  author={Vignac, Clement and Krawczuk, Igor and Siraudin, Antoine and Wang, Bohan and Cevher, Volkan and Frossard, Pascal},
  journal={arXiv preprint arXiv:2209.14734},
  year={2022}
}

@article{rafailov2023direct,
  title={Direct preference optimization: Your language model is secretly a reward model},
  author={Rafailov, Rafael and Sharma, Archit and Mitchell, Eric and Manning, Christopher D and Ermon, Stefano and Finn, Chelsea},
  journal={Advances in neural information processing systems},
  volume={36},
  pages={53728--53741},
  year={2023}
}

@inproceedings{wallace2024diffusion,
  title={Diffusion model alignment using direct preference optimization},
  author={Wallace, Bram and Dang, Meihua and Rafailov, Rafael and Zhou, Linqi and Lou, Aaron and Purushwalkam, Senthil and Ermon, Stefano and Xiong, Caiming and Joty, Shafiq and Naik, Nikhil},
  booktitle={2024 IEEE/CVF Conference on Computer Vision and Pattern Recognition (CVPR)},
  pages={8228--8238},
  year={2024},
  organization={IEEE}
}

@article{borso2025preference,
  title={Preference-based alignment of discrete diffusion models},
  author={Borso, Umberto and Paglieri, Davide and Wells, Jude and Rockt{\"a}schel, Tim},
  journal={arXiv preprint arXiv:2503.08295},
  year={2025}
}

@article{lou2023discrete,
  title={Discrete diffusion modeling by estimating the ratios of the data distribution},
  author={Lou, Aaron and Meng, Chenlin and Ermon, Stefano},
  journal={arXiv preprint arXiv:2310.16834},
  year={2023}
}

@article{nie2026large,
  title={Large language diffusion models},
  author={Nie, Shen and Zhu, Fengqi and You, Zebin and Zhang, Xiaolu and Ou, Jingyang and Hu, Jun and Zhou, Jun and Lin, Yankai and Wen, Ji-Rong and Li, Chongxuan},
  journal={Advances in Neural Information Processing Systems},
  volume={38},
  pages={50608--50646},
  year={2026}
}

@inproceedings{kapur2025diffusion,
  title={Diffusion on syntax trees for program synthesis},
  author={Kapur, Shreyas and Jenner, Erik and Russell, Stuart},
  booktitle={International Conference on Learning Representations},
  volume={2025},
  pages={41414--41432},
  year={2025}
}

@inproceedings{black2024training,
  title={Training diffusion models with reinforcement learning},
  author={Black, Kevin and Janner, Michael and Du, Yilun and Kostrikov, Ilya and Levine, Sergey},
  booktitle={International Conference on Learning Representations},
  volume={2024},
  pages={4965--4987},
  year={2024}
}

@inproceedings{wang2025fine,
  title={Fine-tuning discrete diffusion models via reward optimization with applications to dna and protein design},
  author={Wang, Chenyu and Uehara, Masatoshi and He, Yichun and Wang, Amy and Lal, Avantika and Jaakkola, Tommi and Levine, Sergey and Regev, Aviv and Wang, Hanchen and Biancalani, Tommaso},
  booktitle={International Conference on Learning Representations},
  volume={2025},
  pages={47871--47899},
  year={2025}
}

@inproceedings{zhu2026llada,
  title={Llada 1.5: Variance-reduced preference optimization for large language diffusion models},
  author={Zhu, Fengqi and Wang, Rongzhen and Nie, Shen and Zhang, Xiaolu and Wu, Chunwei and Zhou, Jun and Lin, Yankai and Wen, Ji-Rong and Li, Chongxuan},
  booktitle={Proceedings of the 64th Annual Meeting of the Association for Computational Linguistics (Volume 1: Long Papers)},
  pages={11425--11460},
  year={2026}
}

@article{kakushadze2016101,
  title={101 formulaic alphas},
  author={Kakushadze, Zura},
  journal={Wilmott},
  volume={2016},
  number={84},
  pages={72--81},
  year={2016},
  publisher={Wiley Online Library}
}

@inproceedings{cui2021alphaevolve,
  title={Alphaevolve: A learning framework to discover novel alphas in quantitative investment},
  author={Cui, Can and Wang, Wei and Zhang, Meihui and Chen, Gang and Luo, Zhaojing and Ooi, Beng Chin},
  booktitle={Proceedings of the 2021 International conference on management of data},
  pages={2208--2216},
  year={2021}
}

@article{koza1994genetic,
  title={Genetic programming as a means for programming computers by natural selection},
  author={Koza, John R},
  journal={Statistics and computing},
  volume={4},
  number={2},
  pages={87--112},
  year={1994},
  publisher={Springer}
}

@article{zhao2025quantfactor,
  title={QuantFactor REINFORCE: mining steady formulaic alpha factors with variance-bounded REINFORCE},
  author={Zhao, Junjie and Zhang, Chengxi and Qin, Min and Yang, Peng},
  journal={IEEE Transactions on Signal Processing},
  volume={73},
  pages={2448--2463},
  year={2025},
  publisher={IEEE}
}

@article{shi2025navigating,
  title={Navigating the alpha jungle: an LLM-Powered MCTS framework for formulaic factor mining},
  author={Shi, Yu and Duan, Yitong and Li, Jian},
  journal={arXiv preprint arXiv:2505.11122},
  year={2025}
}

@article{murtagh1991multilayer,
  title={Multilayer perceptrons for classification and regression},
  author={Murtagh, Fionn},
  journal={Neurocomputing},
  volume={2},
  number={5-6},
  pages={183--197},
  year={1991},
  publisher={Elsevier}
}

@inproceedings{chen2016xgboost,
  title={Xgboost: A scalable tree boosting system},
  author={Chen, Tianqi and Guestrin, Carlos},
  booktitle={Proceedings of the 22nd acm sigkdd international conference on knowledge discovery and data mining},
  pages={785--794},
  year={2016}
}

@article{fama1993common,
  title={Common risk factors in the returns on stocks and bonds},
  author={Fama, Eugene F and French, Kenneth R},
  journal={Journal of financial economics},
  volume={33},
  number={1},
  pages={3--56},
  year={1993},
  publisher={Elsevier}
}

@article{bengio2021flow,
  title={Flow network based generative models for non-iterative diverse candidate generation},
  author={Bengio, Emmanuel and Jain, Moksh and Korablyov, Maksym and Precup, Doina and Bengio, Yoshua},
  journal={Advances in neural information processing systems},
  volume={34},
  pages={27381--27394},
  year={2021}
}

@inproceedings{kusner2017grammar,
  title={Grammar variational autoencoder},
  author={Kusner, Matt J and Paige, Brooks and Hern{\'a}ndez-Lobato, Jos{\'e} Miguel},
  booktitle={International conference on machine learning},
  pages={1945--1954},
  year={2017},
  organization={PMLR}
}

@article{hoerl1970ridge,
  title={Ridge regression: Biased estimation for nonorthogonal problems},
  author={Hoerl, Arthur E and Kennard, Robert W},
  journal={Technometrics},
  volume={12},
  number={1},
  pages={55--67},
  year={1970},
  publisher={Taylor \& Francis}
}

@inproceedings{chen2021empirical,
  title={An empirical study of financial factor mining based on gene expression programming},
  author={Chen, Tianxiang and Chen, Wei and Du, Luyao},
  booktitle={2021 4th International Conference on Advanced Electronic Materials, Computers and Software Engineering (AEMCSE)},
  pages={1113--1117},
  year={2021},
  organization={IEEE}
}
\bibliographystyle{iclr2027/iclr2027_conference}

\clearpage
\appendix

\setcounter{secnumdepth}{3}
\section{Notation}
Table~\ref{tab:notation} summarizes the principal notation used throughout the paper, and Appendix~\ref{app:features_ops} lists the formula vocabulary.

\begin{table}[h]
\centering
\caption{Summary of the principal notation used in \shortname{}.}
\label{tab:notation}
\footnotesize
\renewcommand{\arraystretch}{1.08}
\begin{tabularx}{\linewidth}{>{\raggedright\arraybackslash}p{3.3cm}X}
\toprule
\textbf{Symbol} & \textbf{Definition} \\
\midrule
\multicolumn{2}{l}{\textit{Data, pools, and evaluation}} \\
$X_d$, $y_d$ & Market panel and forward returns on date $d$ \\
$\mathcal X$, $f$ & Formula space defined by the vocabulary $\mathcal V$ and a formula in it \\
$F=[f_1,\ldots,f_K]$, $K$ & Alpha pool and pool size \\
$\mathcal M_K(\mathcal X)$ & Size-$K$ pools of canonically distinct formulas \\
$\widehat w_F$, $\gamma$ & Ridge combination weights and penalty (Eq.~\ref{eq:ridge}) \\
$U_S(F)$ & Utility of pool $F$ on split $S$: mean daily Spearman correlation of the combined prediction \\
$\mathcal A$, $\mathcal E$, $B$ & Submitted canonical formulas, eligible formulas, and budget on $|\mathcal A|$ \\
\midrule
\multicolumn{2}{l}{\textit{Hierarchical diffusion}} \\
$T$, $t$, $\rho_t$ & Number of diffusion levels, level index, and masking schedule \\
$U_u$, $\Phi_\rho$ & Clock of absorbable node $u$ and subtree-masking operator \\
$F_t$, $\mathbf m$, $H_t$ & Pool corrupted to level $t$, mask token, and exposed masks \\
$a_t$, $w_t$ & Mask-retention and reveal probabilities \\
$\mu_{t-1}^{u,F_0}$ & Level-$(t-1)$ distribution of the subtree at $u$ \\
$J_t$, $u_{t,r}$, $G^{(r)}$ & Reveal set, its $r$-th position, and the updated pool \\
$\psi$, $\mathbf c_{G,u,t}$, $E_t$ & Formula encoder, pool context, and level embedding \\
$r_\theta$, $\theta_{\mathrm{ref}}$ & Reverse model and frozen reference model \\
$\alpha$, $t_0$ & Edit strength and starting level of refinement \\
\midrule
\multicolumn{2}{l}{\textit{Reward and buffer}} \\
$c_{ij}$, $d_i$, $\delta$ & Pairwise redundancy, nearest-member redundancy, and tolerated similarity \\
$R_{\mathrm{outer}}$, $R_{\mathrm{inner}}$, $R$ & Outer predictive, inner diversity, and pool rewards \\
$w_{\mathrm{inner}}$ & Weight of the inner diversity reward \\
$\mathcal Q$, $C_{\mathcal Q}$, $M$ & Elite buffer, its capacity, and number of initial pools \\
$P_{\mathcal Q}$, $T_{\mathrm r}$ & Buffer sampling distribution and its temperature \\
$\mathcal C$, $\widehat F$ & Validation archive and selected final pool \\
\midrule
\multicolumn{2}{l}{\textit{Training and search}} \\
$\widehat{\mathcal L}_{\mathrm{VB}}$, $\ell_\theta$ & Reconstruction objective and single-pool loss \\
$\mathcal L_{\mathrm{pref}}$, $\mathcal P_{\mathcal Q}$ & Preference objective and distribution of accepted pairs \\
$\beta$, $m_{\mathrm{dpo}}$ & Preference strength and reward margin \\
$\lambda_{\mathrm{VB}}$ & Weight of the reconstruction objective \\
$N_{\mathrm c}$, $N$ & Proposal batch size and training minibatch size \\
$S_0$, $S_{\mathrm u}$, $N_{\mathrm{stall}}$ & Initialization updates, updates per round, and stall limit \\
\bottomrule
\end{tabularx}
\end{table}

\subsection{Features and Operators}
\label{app:features_ops}
Table~\ref{tab:dsl_vocabulary} lists the features, operators, windows,
constants, and special tokens that define the formula space $\mathcal X$.

\begin{table}[h]
\centering
\caption{Formula vocabulary. The 61 semantic tokens in $\mathcal V$ define the formula space $\mathcal X$, and the four special tokens control generation and diffusion. In the stack effect, $w$ denotes a window token.}
\label{tab:dsl_vocabulary}
\small
\renewcommand{\arraystretch}{1.12}
\begin{tabularx}{\linewidth}{@{}l>{\raggedright\arraybackslash}Xl@{}}
\toprule
\textbf{Type (count)} & \textbf{Tokens} & \textbf{Stack effect} \\
\midrule
Feature (6) & \texttt{Open}, \texttt{Close}, \texttt{High}, \texttt{Low}, \texttt{Vwap}, \texttt{Volume} & push \\
Constant (14) & $\pm30$, $\pm10$, $\pm5$, $\pm2$, $\pm1$, $\pm0.5$, $\pm0.01$ & push \\
Window (5) & $1$, $5$, $10$, $20$, $40$ & push window \\
Unary (6) & \texttt{Abs}, \texttt{Slog1p}, \texttt{Inv}, \texttt{Sign}, \texttt{Log}, \texttt{Rank} & $1\to1$ \\
Binary (7) & \texttt{Add}, \texttt{Sub}, \texttt{Mul}, \texttt{Div}, \texttt{Pow}, \texttt{Greater}, \texttt{Less} & $2\to1$ \\
Temporal (21) & \texttt{Ref}, \texttt{TsMean}, \texttt{TsSum}, \texttt{TsStd}, \texttt{TsIr}, \texttt{TsMinMaxDiff}, \texttt{TsMaxDiff}, \texttt{TsMinDiff}, \texttt{TsVar}, \texttt{TsSkew}, \texttt{TsKurt}, \texttt{TsMax}, \texttt{TsMin}, \texttt{TsMed}, \texttt{TsMad}, \texttt{TsRank}, \texttt{TsDelta}, \texttt{TsDiv}, \texttt{TsPctChange}, \texttt{TsWMA}, \texttt{TsEMA} & $1{+}w\to1$ \\
Pairwise temporal (2) & \texttt{TsCov}, \texttt{TsCorr} & $2{+}w\to1$ \\
\midrule
Special (4) & \textsc{beg}, \textsc{sep}, \textsc{mask}, \textsc{fill} & control \\
\bottomrule
\end{tabularx}
\end{table}

\section{Supplementary Model Details}
\label{app:model_supplement}
This section provides the theoretical analysis of \shortname{}
(Appendix~\ref{app:theoretical_analysis}), its complete search and
training procedure (Appendix~\ref{app:algorithm}), and its computational
complexity (Appendix~\ref{app:complexity}).

\subsection{Theoretical Derivations}
\label{app:theoretical_analysis}
We analyze the probabilistic structure of the two learned components of \shortname{}. Appendix~\ref{app:theory_diffusion} derives the posterior of the hierarchical corruption process and the reverse transition it induces, corresponding to Section~\ref{sec:pool_diffusion}. Appendix~\ref{app:theory_objective} relates the reconstruction and preference objectives of Section~\ref{sec:reward_guide} to a variational bound and to a KL-regularized pool-level reward maximization, and identifies the approximations involved.

For a clean pool $F_0$, let $\mathcal N(F_0)$ be the set of absorbable AST nodes of its $K$ formulas. Window tokens are attached to their temporal operators and are not absorbable. For $u\in\mathcal N(F_0)$, $\operatorname{Sub}(u)$ contains $u$ and its absorbable descendants, and $\operatorname{Desc}(u)=\operatorname{Sub}(u)\setminus\{u\}$. For an antichain $H\subseteq\mathcal N(F_0)$, in which no node is an ancestor of another, $\Phi_H(F_0)$ replaces each $\operatorname{Sub}(u)$ with $u\in H$ by $\mathbf m$. With the clocks $\{U_u\}$ of Section~\ref{sec:pool_diffusion}, the visible, active, and exposed sets and the schedule quantities are
\begin{equation}
\begin{aligned}
V(H)&=\mathcal N(F_0)\setminus\textstyle\bigcup_{u\in H}\operatorname{Sub}(u),
&\qquad
A_\rho&=\{u\in\mathcal N(F_0):U_u\le\rho\},\\
H_\rho&=\{u\in A_\rho:\text{no ancestor of $u$ lies in $A_\rho$}\},
&\qquad
\Phi_\rho(F_0)&=\Phi_{H_\rho}(F_0),\\
\eta_t&=\frac{\rho_t-\rho_{t-1}}{1-\rho_{t-1}},
&\qquad
a_t&=\frac{\rho_{t-1}}{\rho_t},\quad w_t=\frac{\rho_t-\rho_{t-1}}{\rho_t},
\end{aligned}
\label{eq:app_setup}
\end{equation}
where $0=\rho_0<\rho_1<\cdots<\rho_T=1$ and $H_t:=H_{\rho_t}$ collects the exposed masks of $F_t=\Phi_{\rho_t}(F_0)$. A corrupted pool is treated as a partially masked AST, so $H_t$ is recoverable from $F_t$. For the $r$-th reconstruction of a transition, we abbreviate $r_\vartheta^{(r)}=r_\vartheta(Y_{u_r}\mid G^{(r-1)},u_r,t)$ for $\vartheta\in\{\theta,\theta_{\mathrm{ref}}\}$. We adopt two conventions. (A1)~$r_\vartheta$ is renormalized over grammar-valid actions, assigns positive probability to every valid target, including targets with \textsc{mask} at descendant positions, and never outputs the bare target $\mathbf m$ at the queried position. (A2)~The order $\pi$ is drawn from a parameter-free distribution $\nu(\pi\mid J)$ over the set $\mathcal O(J)$ of interleavings of $J$ that preserve the RPN order within each formula. Canonicalization, duplicate replacement, and executability filtering follow reverse generation and lie outside the analyzed model.

\subsubsection{Pool-Aware Hierarchical Diffusion}
\label{app:theory_diffusion}

\begin{proposition}[Hierarchical corruption posterior]
\label{prop:app_posterior}
For every $t\in\{0,\ldots,T\}$ and every antichain $H\subseteq\mathcal N(F_0)$,
\begin{equation}
q\bigl(F_t=\Phi_H(F_0)\mid F_0\bigr)=\rho_t^{|H|}\,(1-\rho_t)^{|V(H)|}.
\label{eq:app_marginal}
\end{equation}
For $t\ge1$, let $J_t=H_t\setminus H_{t-1}$ and $S_t=H_t\cap H_{t-1}$ be the revealed and retained masks. For $u\in J_t$, let $\mathcal H^u=H_{t-1}\cap\operatorname{Desc}(u)$ and $V^u=\operatorname{Sub}(u)\setminus\bigcup_{v\in\mathcal H^u}\operatorname{Sub}(v)$. The posterior factorizes over the exposed masks as
\begin{equation}
q(F_{t-1}\mid F_t,F_0)
=\prod_{u\in H_t}\Bigl[a_t\,\delta_{\mathbf m}(Y_u)+w_t\,\mu_{t-1}^{u,F_0}(Y_u)\Bigr],
\label{eq:app_posterior}
\end{equation}
where $\mu_{t-1}^{u,F_0}(Y_u)=\rho_{t-1}^{|\mathcal H^u|}(1-\rho_{t-1})^{|V^u|-1}$.
\end{proposition}

\begin{proof}
The exposed set equals $H$ if and only if every node of $H$ is active and every node of $V(H)$ is inactive. Under these conditions, each active node lies below some $u\in H$ whose ancestors are all inactive. Conversely, an active node in $V(H)$ would have its highest active ancestor in $H_\rho=H$, which contradicts its membership in $V(H)$. Clocks strictly below $H$ do not enter, and the independence of the clocks gives Eq.~\ref{eq:app_marginal}. In particular, $\rho_T=1$ forces $F_T=(\mathbf m,\ldots,\mathbf m)$. Because the clocks are shared and $\rho_t$ increases, $A_{\rho_{t-1}}\subseteq A_{\rho_t}$. Given $F_{t-1}$, the clocks on $V(H_{t-1})$ are uniform on $(\rho_{t-1},1]$, so each visible node is newly masked with probability $\eta_t$, independently, while nodes inside masked subtrees cannot affect $H_t$. The forward process is therefore Markov with $q(F_t\mid F_{t-1})=\eta_t^{|J_t|}(1-\eta_t)^{|V(H_t)|}$. Each element of $H_{t-1}$ lies either in $S_t$ or below exactly one $u\in J_t$, which gives
\begin{equation}
V(H_{t-1})=V(H_t)\sqcup\textstyle\bigsqcup_{u\in J_t}V^u,
\qquad
|H_{t-1}|=|S_t|+\textstyle\sum_{u\in J_t}|\mathcal H^u|,
\qquad
|H_t|=|S_t|+|J_t|.
\label{eq:app_decomposition}
\end{equation}
Applying Bayes' rule with Eqs.~\ref{eq:app_marginal} and~\ref{eq:app_decomposition}, together with $(1-\eta_t)(1-\rho_{t-1})=1-\rho_t$ and $\eta_t(1-\rho_{t-1})=w_t\rho_t$, yields
\begin{equation}
\begin{aligned}
q(F_{t-1}\mid F_t,F_0)
&=\frac{q(F_t\mid F_{t-1})\,q(F_{t-1}\mid F_0)}{q(F_t\mid F_0)}
=\frac{\eta_t^{|J_t|}(1-\eta_t)^{|V(H_t)|}\,\rho_{t-1}^{|H_{t-1}|}(1-\rho_{t-1})^{|V(H_{t-1})|}}
{\rho_t^{|H_t|}(1-\rho_t)^{|V(H_t)|}}\\
&=\Bigl(\frac{\rho_{t-1}}{\rho_t}\Bigr)^{|S_t|}
\prod_{u\in J_t}\frac{\eta_t(1-\rho_{t-1})}{\rho_t}\,
\rho_{t-1}^{|\mathcal H^u|}(1-\rho_{t-1})^{|V^u|-1}\\
&=a_t^{|S_t|}\prod_{u\in J_t}w_t\,\mu_{t-1}^{u,F_0}(Y_u),
\end{aligned}
\label{eq:app_bayes}
\end{equation}
where the factors on $V(H_t)$ cancel and one factor $1-\rho_{t-1}$ from each $V^u$ is absorbed into $w_t$. The distribution $\mu_{t-1}^{u,F_0}$ is Eq.~\ref{eq:app_marginal} applied to the forest $\operatorname{Desc}(u)$ at level $t-1$ and is therefore normalized. By (A1), it never produces the bare $\mathbf m$, so the two terms in Eq.~\ref{eq:app_posterior} have disjoint supports, and expanding the product recovers Eq.~\ref{eq:app_bayes}.
\end{proof}

\begin{remark}
At $t=1$, we have $a_1=0$, $w_1=1$, and $\mu_0^{u,F_0}$ is the point mass on the clean subtree of $u$, so the last reverse transition deterministically completes every remaining mask. Drawing one set of clocks and reading off $(\Phi_{\rho_t}(F_0),\Phi_{\rho_{t-1}}(F_0))$ yields an exact sample of $q(F_t,F_{t-1}\mid F_0)$, in which each $u\in H_t$ enters $J_t$ independently with probability $w_t$ and independently of $F_0$. The reverse process in Section~\ref{sec:pool_diffusion} uses the same Bernoulli selection, and the pairs $(J_b,Y_{b,r})$ in Eq.~\ref{eq:variational_diffusion} are sampled in this way. The hierarchy enters the posterior only through $\mu_{t-1}^{u,F_0}$: a revealed position receives a complete subtree whose descendants remain corrupted at level $t-1$, so a single transition can edit a leaf, a subexpression, or a whole formula.
\end{remark}

\begin{proposition}[Pool-aware reverse transition]
\label{prop:app_reverse}
Let $q(F_0)$ be the distribution of clean pools, such as $P_{\mathcal Q}$, and let $h_r=(F_t,u_{1:r},Y_{u_1},\ldots,Y_{u_{r-1}})$ denote the reconstruction history. For every order $\pi=(u_1,\ldots,u_{|J_t|})\in\mathcal O(J_t)$, the true reverse transition and the reverse model in Section~\ref{sec:pool_diffusion} are
\begin{equation}
\begin{aligned}
q(F_{t-1}\mid F_t)&=a_t^{|S_t|}\,w_t^{|J_t|}\prod_{r=1}^{|J_t|}\mathsf Q_t(Y_{u_r}\mid h_r),\\
p_\theta(F_{t-1}\mid F_t)&=a_t^{|S_t|}\,w_t^{|J_t|}\sum_{\pi\in\mathcal O(J_t)}\nu(\pi\mid J_t)\prod_{r=1}^{|J_t|}r_\theta^{(r)},
\end{aligned}
\label{eq:app_reverse}
\end{equation}
where $\mathsf Q_t(Y_{u_r}\mid h_r)=\mathbb E_{F_0\sim q(\cdot\mid F_t,Y_{u_{<r}})}\bigl[\mu_{t-1}^{u_r,F_0}(Y_{u_r})\bigr]$ is the posterior predictive of the $r$-th target. Let $\mathsf Q_t(\cdot\mid F_t)$ and $\mathsf P_\theta(\cdot\mid F_t,J_t)$ denote the content distributions over $Y_{J_t}$ in the two lines. The transition divergence then reduces to the reconstructed content,
\begin{equation}
\operatorname{KL}\bigl(q(F_{t-1}\mid F_t)\,\big\|\,p_\theta(F_{t-1}\mid F_t)\bigr)
=\mathbb E_{J_t}\Bigl[\operatorname{KL}\bigl(\mathsf Q_t(\cdot\mid F_t)\,\big\|\,\mathsf P_\theta(\cdot\mid F_t,J_t)\bigr)\Bigr],
\label{eq:app_reverse_kl}
\end{equation}
where each $u\in H_t$ enters $J_t$ independently with probability $w_t$.
\end{proposition}

\begin{proof}
When Eq.~\ref{eq:app_posterior} is marginalized over $q(F_0\mid F_t)$, the selection factor does not depend on $F_0$, and the chain rule applies to the remaining mixture for any order $\pi$:
\begin{equation}
\begin{aligned}
q(F_{t-1}\mid F_t)
&=\sum_{F_0}q(F_0\mid F_t)\,q(F_{t-1}\mid F_t,F_0)\\
&=a_t^{|S_t|}\,w_t^{|J_t|}\sum_{F_0}q(F_0\mid F_t)\prod_{u\in J_t}\mu_{t-1}^{u,F_0}(Y_u)\\
&=a_t^{|S_t|}\,w_t^{|J_t|}\prod_{r=1}^{|J_t|}\sum_{F_0}q(F_0\mid F_t,Y_{u_{<r}})\,\mu_{t-1}^{u_r,F_0}(Y_{u_r}),
\end{aligned}
\label{eq:app_true_reverse}
\end{equation}
where $q(F_0\mid F_t,Y_{u_{<r}})\propto q(F_0\mid F_t)\prod_{s<r}\mu_{t-1}^{u_s,F_0}(Y_{u_s})$. The model transition follows from the reverse procedure. $J_t$ is selected with probability $w_t$ per position, $\pi\sim\nu(\cdot\mid J_t)$, and each target is drawn from $r_\theta^{(r)}$ and spliced into $G^{(r)}=\operatorname{Splice}(G^{(r-1)},u_r,Y_{u_r})$ before the next prediction, with $G^{(0)}=F_t$. By (A1), $J_t$ is identifiable from $F_{t-1}$, so both transitions assign the same probability $a_t^{|S_t|}w_t^{|J_t|}$ to the selection. Summing $q\log(q/p_\theta)$ over $F_{t-1}$, first over the contents and then over the selections, gives Eq.~\ref{eq:app_reverse_kl}.
\end{proof}

\begin{remark}
The content distribution $\mathsf Q_t(Y_{J_t}\mid F_t)=\mathbb E_{F_0\sim q(\cdot\mid F_t)}\prod_{u\in J_t}\mu_{t-1}^{u,F_0}(Y_u)$ is a mixture of products. It does not factorize across positions whenever $q(F_0\mid F_t)$ couples the pool members, as it does for buffered pools that are selected jointly by the non-separable reward in Eq.~\ref{eq:pool_reward}. A reverse model that reconstructs each position independently from $F_t$ therefore cannot represent it. The sequential model instead conditions each target on the updated pool $G^{(r-1)}$. If $G^{(r-1)}$ retains $h_r$ and $r_\theta^{(r)}=\mathsf Q_t(\cdot\mid h_r)$, every order yields the same product, and $p_\theta(F_{t-1}\mid F_t)=q(F_{t-1}\mid F_t)$ for any $\nu$. This is the formal role of pool-aware reverse generation. Candidate generation composes these kernels from the partial corruption level $t_0$,
\begin{equation}
P_\theta(F'\mid F)
=\sum_{F_{t_0}}q(F_{t_0}\mid F)\sum_{F_1,\ldots,F_{t_0-1}}\;\prod_{t=1}^{t_0}p_\theta(F_{t-1}\mid F_t)\Big|_{F_0=F'},
\label{eq:app_candidate}
\end{equation}
which is the pool-conditional generator in Eq.~\ref{eq:pool_aware_generation}.
\end{remark}

\subsubsection{Reward-Guided Pool Refinement and Training Objective}
\label{app:theory_objective}

Let $\xi=(t,\{U_u\},\pi)$ collect the diffusion level, the clocks, and the reconstruction order, and let $\Omega=(F_{1:T},\pi_{1:T})$ denote an augmented reverse path. For $\vartheta\in\{\theta,\theta_{\mathrm{ref}}\}$, the unnormalized and normalized single-pool reconstruction losses are
\begin{equation}
\Lambda_\vartheta(F_0;\xi)=-T\sum_{r=1}^{|J_t|}\log r_\vartheta^{(r)},
\qquad
\ell_\vartheta(F_0;\xi)=-\sum_{r=1}^{|J_t|}\frac{\log r_\vartheta^{(r)}}{n\,L_r},
\qquad
n=\max\{1,|J_t|\},
\label{eq:app_losses}
\end{equation}
so that $\widehat{\mathcal L}_{\mathrm{VB}}=(T/N)\sum_{b=1}^{N}\ell_\theta(F_0^{(b)};\xi_b)$ in Eq.~\ref{eq:variational_diffusion}. We write $\mathcal L_{\mathrm{VB}}(F_0;\vartheta)=\mathbb E_\xi[\Lambda_\vartheta(F_0;\xi)]$ and let $\bar L$ be the maximal target length, including \textsc{sep}.

\begin{proposition}[Variational bound and normalized surrogate]
\label{prop:app_elbo}
There exists $\mathcal H(F_0)\ge0$, independent of $\theta$, such that
\begin{equation}
-\log p_\theta(F_0)
\le\mathcal L_{\mathrm{VB}}(F_0;\theta)-\mathcal H(F_0)
\le T\bar L\,\mathbb E_\xi\bigl[n_\xi\,\ell_\theta(F_0;\xi)\bigr]-\mathcal H(F_0).
\label{eq:app_elbo}
\end{equation}
Moreover, for a fixed history $h$ with $\mathsf Q=\mathsf Q_t(\cdot\mid h)$, the minimizer of $\mathbb E_\xi[\ell_\theta]$ over unrestricted conditionals $r(\cdot\mid h)$ is the length-tilted distribution
\begin{equation}
\begin{gathered}
r^\star(Y\mid h)=\frac{\mathsf Q(Y)}{L(Y)\,Z(h)},
\qquad
Z(h)=\mathbb E_{\mathsf Q}\bigl[L(Y)^{-1}\bigr],\\
\max\bigl\{\operatorname{KL}(\mathsf Q\,\|\,r^\star),\,\operatorname{KL}(r^\star\,\|\,\mathsf Q)\bigr\}\le\log\bar L.
\end{gathered}
\label{eq:app_tilt}
\end{equation}
\end{proposition}

\begin{proof}
The forward path law factorizes as $q(\Omega\mid F_0)=q(F_T\mid F_0)\prod_{t=1}^{T}q(F_{t-1}\mid F_t,F_0)\,\nu(\pi_t\mid J_t)$, and the reverse path law as $p_\theta(F_0,\Omega)=p(F_T)\prod_{t=1}^{T}p_\theta(F_{t-1},\pi_t\mid F_t)$, where $p_\theta(F_{t-1},\pi_t\mid F_t)$ is the summand in Eq.~\ref{eq:app_reverse}. Since $\rho_T=1$, $q(F_T\mid F_0)=p(F_T)$. Jensen's inequality and Eqs.~\ref{eq:app_posterior} and~\ref{eq:app_reverse}, in which the selection factors and $\nu$ cancel, give
\begin{equation}
\begin{aligned}
-\log p_\theta(F_0)
&\le\mathbb E_{q(\Omega\mid F_0)}\Bigl[\sum_{t=1}^{T}\log\frac{q(F_{t-1}\mid F_t,F_0)\,\nu(\pi_t\mid J_t)}{p_\theta(F_{t-1},\pi_t\mid F_t)}\Bigr]\\
&=\sum_{t=1}^{T}\mathbb E_{q}\Bigl[\sum_{u\in J_t}\log\mu_{t-1}^{u,F_0}(Y_u)-\sum_{r=1}^{|J_t|}\log r_\theta^{(r)}\Bigr]\\
&=T\,\mathbb E_{t\sim\operatorname{Unif}\{1,\ldots,T\}}\,\mathbb E_{q}\Bigl[-\sum_{r=1}^{|J_t|}\log r_\theta^{(r)}\Bigr]-\mathcal H(F_0)
=\mathcal L_{\mathrm{VB}}(F_0;\theta)-\mathcal H(F_0),
\end{aligned}
\label{eq:app_elbo_proof}
\end{equation}
where $\mathcal H(F_0)=\sum_{t}\mathbb E_q[-\sum_{u\in J_t}\log\mu_{t-1}^{u,F_0}(Y_u)]\ge0$. Since $-\log r_\theta^{(r)}\ge0$ and $1\le L_r\le\bar L$, every $\xi$ satisfies $\Lambda_\theta(F_0;\xi)\le T\bar L\,n_\xi\,\ell_\theta(F_0;\xi)$, and taking expectations gives the second inequality of Eq.~\ref{eq:app_elbo}. For Eq.~\ref{eq:app_tilt}, the weight $1/n$ depends only on $J_t$, which is independent of $Y_{u_r}$ given $h$. It therefore scales the contribution of $h$ by $c(h)=\mathbb E[1/n\mid h]>0$. The weight $1/L(Y)$ depends on the target itself, and
\begin{equation}
\begin{aligned}
c(h)\,\mathbb E_{\mathsf Q}\Bigl[\frac{-\log r(Y\mid h)}{L(Y)}\Bigr]
&=c(h)\,Z(h)\,\mathbb E_{r^\star}\bigl[-\log r(Y\mid h)\bigr]\\
&=c(h)\,Z(h)\Bigl[\mathbb H(r^\star)+\operatorname{KL}\bigl(r^\star\,\|\,r(\cdot\mid h)\bigr)\Bigr],
\end{aligned}
\label{eq:app_tilt_proof}
\end{equation}
which is minimized at $r=r^\star$, where $\mathbb H$ denotes the Shannon entropy. Finally, $\log(r^\star/\mathsf Q)=-\log L(Y)-\log Z(h)\in[-\log\bar L,\log\bar L]$ because $1\le L(Y)\le\bar L$ and $\bar L^{-1}\le Z(h)\le1$, which bounds both divergences.
\end{proof}

\begin{remark}
Eq.~\ref{eq:app_elbo} explains why Eq.~\ref{eq:variational_diffusion} is a surrogate rather than an evidence lower bound. The unnormalized loss $\Lambda_\theta$ is an unbiased single-sample estimator of the $\theta$-dependent part of the bound. Teacher forcing is exact because the bound is evaluated along paths drawn from $q$, and the factor $T$ arises from sampling the level uniformly. The normalized loss controls the same bound up to the known factor $T\bar L\,n_\xi$. By Eq.~\ref{eq:app_tilt}, the factor $1/n$ equalizes pools with different numbers of reveals without changing the minimizer. The factor $1/L_r$ equalizes leaf-, subtree-, and alpha-level targets at the cost of a tilt of at most $\log\bar L$ nats toward shorter completions. As a result, $-\ell_\theta$ is a per-target, per-token reconstruction score that is comparable across pools of different structure, which is the property used by the preference objective.
\end{remark}

\begin{proposition}[Pool-level preference objective]
\label{prop:app_pref}
Consider the KL-regularized pool-level reward maximization over augmented reverse paths,
\begin{equation}
\max_{\theta}\;\mathbb E_{F_0\sim p_\theta}\bigl[R(F_0)\bigr]-\beta'\,\operatorname{KL}\bigl(p_\theta(F_0,\Omega)\,\big\|\,p_{\mathrm{ref}}(F_0,\Omega)\bigr),
\qquad\beta'>0.
\label{eq:app_kl_objective}
\end{equation}
Assume that its maximizer is attained within the model class, and model preferences as $\Pr(F_{\mathrm w}\succ F_{\mathrm l})=\sigma\bigl(R(F_{\mathrm w})-R(F_{\mathrm l})\bigr)$. Then the negative log-likelihood of accepted pairs is
\begin{equation}
\begin{aligned}
\mathcal L^\star_{\mathrm{pref}}(\theta)&=-\mathbb E_{\mathcal P_{\mathcal Q}}\log\sigma\bigl(\beta'\Delta^\star_\theta\bigr),\\
\Delta^\star_\theta
&=\bigl[\mathcal L^{\mathrm l}_{\mathrm{VB}}(\theta)-\mathcal L^{\mathrm w}_{\mathrm{VB}}(\theta)\bigr]
-\bigl[\mathcal L^{\mathrm l}_{\mathrm{VB}}(\theta_{\mathrm{ref}})-\mathcal L^{\mathrm w}_{\mathrm{VB}}(\theta_{\mathrm{ref}})\bigr],
\end{aligned}
\label{eq:app_exact_pref}
\end{equation}
where $\mathcal L^{a}_{\mathrm{VB}}(\vartheta)=\mathcal L_{\mathrm{VB}}(F_a;\vartheta)$ for $a\in\{\mathrm w,\mathrm l\}$. Let $X$ be the single-sample estimate of $\beta'\Delta^\star_\theta$ obtained from $\Lambda_\vartheta$ under a coupling $\Gamma$ of $(\xi_{\mathrm w},\xi_{\mathrm l})$ that preserves the corruption law of each pool. Then
\begin{equation}
-\log\sigma\bigl(\mathbb E_\Gamma X\bigr)
\le\mathbb E_\Gamma\bigl[-\log\sigma(X)\bigr]
\le-\log\sigma\bigl(\mathbb E_\Gamma X\bigr)+\tfrac18\operatorname{Var}_\Gamma(X),
\label{eq:app_jensen}
\end{equation}
and replacing $\Lambda_\vartheta$ by $T\ell_\vartheta$ with $\beta=T\beta'$ yields Eq.~\ref{eq:pool_preference}.
\end{proposition}

\begin{proof}
Eq.~\ref{eq:app_kl_objective} equals $-\beta'\operatorname{KL}(p_\theta\,\|\,p_{\mathrm{ref}}e^{R/\beta'}/Z)+\beta'\log Z$ with $Z=\mathbb E_{p_{\mathrm{ref}}}[e^{R(F_0)/\beta'}]$, so its maximizer is $p^\star=p_{\mathrm{ref}}e^{R/\beta'}/Z$. The tilt depends only on the endpoint, so $\log(p^\star/p_{\mathrm{ref}})=R(F_0)/\beta'-\log Z$ on every path ending at $F_0$, which by (A1) includes the support of $q(\Omega\mid F_0)$. We parameterize $p^\star$ by $p_\theta$~\citep{rafailov2023direct}, average under $q(\Omega\mid F_0)$, and apply the cancellations in Eq.~\ref{eq:app_elbo_proof}, which also remove $\mathcal H(F_0)$:
\begin{equation}
\begin{aligned}
R(F_0)-\beta'\log Z
&=\beta'\,\mathbb E_{q(\Omega\mid F_0)}\Bigl[\log\frac{p_\theta(F_0,\Omega)}{p_{\mathrm{ref}}(F_0,\Omega)}\Bigr]
=\beta'\,\mathbb E_{q(\Omega\mid F_0)}\Bigl[\sum_{t=1}^{T}\sum_{r=1}^{|J_t|}\log\frac{r_\theta^{(r)}}{r_{\mathrm{ref}}^{(r)}}\Bigr]\\
&=-\beta'\bigl[\mathcal L_{\mathrm{VB}}(F_0;\theta)-\mathcal L_{\mathrm{VB}}(F_0;\theta_{\mathrm{ref}})\bigr].
\end{aligned}
\label{eq:app_reward_identity}
\end{equation}
Substituting Eq.~\ref{eq:app_reward_identity} for $F_{\mathrm w}$ and $F_{\mathrm l}$ into the preference model cancels $\log Z$ and gives Eq.~\ref{eq:app_exact_pref}. For Eq.~\ref{eq:app_jensen}, let $A_a=\beta'[\Lambda_\theta(F_a;\xi_a)-\Lambda_{\mathrm{ref}}(F_a;\xi_a)]$ for $a\in\{\mathrm w,\mathrm l\}$ and $X=A_{\mathrm l}-A_{\mathrm w}$. Because $\Gamma$ preserves the marginal laws, $\mathbb E_\Gamma X=\beta'\Delta^\star_\theta$. The function $g(x)=-\log\sigma(x)$ is convex with $g''(x)=\sigma(x)\sigma(-x)\le1/4$, so Jensen's inequality gives the lower bound, and a second-order expansion around $\mathbb E_\Gamma X$ gives the upper bound. The last claim follows from Eq.~\ref{eq:app_losses}, since $\beta'\,T\ell_\vartheta=\beta\,\ell_\vartheta$.
\end{proof}

\begin{remark}
Since $\operatorname{Var}_\Gamma(X)=\operatorname{Var}(A_{\mathrm l})+\operatorname{Var}(A_{\mathrm w})-2\operatorname{Cov}_\Gamma(A_{\mathrm l},A_{\mathrm w})$, only the covariance depends on the coupling. The matched pseudorandom stream in Section~\ref{sec:reward_guide} therefore tightens Eq.~\ref{eq:app_jensen} whenever it induces positive covariance, and evaluating $\theta$ and $\theta_{\mathrm{ref}}$ under the same $\xi_a$ makes $X\equiv0$ at initialization. The normalization in Eq.~\ref{eq:pool_preference} turns the implicit reward into the per-target, per-token quantity $-\beta(\ell_\theta-\ell_{\mathrm{ref}})$, which prevents pools with more or longer targets from dominating the margin. Proposition~\ref{prop:app_elbo} bounds the resulting deviation from Eq.~\ref{eq:app_exact_pref}.

Although $R$ is computable, the pairwise objective has three advantages over maximizing $\mathbb E_{p_\theta}[R]$ directly. First, the expected-reward gradient requires fresh on-policy samples, each consuming distinct-formula evaluations from the budget $B$ in Eq.~\ref{eq:budget}, whereas Eq.~\ref{eq:pool_preference} reuses scored buffer pools. Second, the pair-conditional likelihood does not involve the distribution from which pairs are drawn, so off-policy elite pairs are admissible. Moreover, $R$ enters only through the acceptance event, which is invariant to positive affine transformations of $R$. Third, with $\mathcal B_{\mathcal Q}\sim P_{\mathcal Q}$, the reconstruction term is a reward-weighted regression toward the Boltzmann distribution over standardized buffer rewards. It reallocates likelihood among retained pools but never contrasts them, and the preference term supplies this contrast. Freezing $\theta_{\mathrm{ref}}$ fixes the target $p^\star\propto p_{\mathrm{ref}}e^{R/\beta'}$. Writing $\Delta$ for the bracket in Eq.~\ref{eq:pool_preference}, the reference enters the gradient only through a per-pair offset that rescales the weight $\sigma(-\beta\Delta)$ without changing the direction $\nabla_\theta\ell^{\mathrm l}_\theta-\nabla_\theta\ell^{\mathrm w}_\theta$. Resetting the reference after each of $k$ optimal rounds would instead yield $p_{\mathrm{ref}}e^{kR/\beta'}$, a progressively colder target that erodes proposal diversity. The proposition does not claim that finite-step training attains $p^\star$. The fixed selection factors make attainability an assumption, and because labels are deterministic given $R$, the preference model is recovered only in its zero-temperature limit. Eq.~\ref{eq:pool_preference} therefore acts as a logistic ranking loss on accepted pairs.
\end{remark}

\subsection{Algorithm Flow Chart}
\label{app:algorithm}
The algorithm flow of \shortname{} is shown in Algorithm~\ref{alg:pool_refinement}.

\begin{algorithm}[t]
\caption{\shortname{} pool refinement and training}
\label{alg:pool_refinement}
\begin{algorithmic}[1]
\Require buffer capacity $C_{\mathcal Q}$, number of initial pools $M\le C_{\mathcal Q}$,
budget $B$ (Eq.~\ref{eq:budget}), proposal batch size $N_{\mathrm c}$, minibatch
size $N$, initialization updates $S_0$, per-round updates $S_{\mathrm u}$,
stall limit $N_{\mathrm{stall}}$
\State Build $\mathcal Q_0$ by Eqs.~\ref{eq:elite_buffer_initialization}
and~\ref{eq:pool_reward}, and add all submitted canonical formulas to $\mathcal A$
\State $\mathcal Q\gets\mathcal Q_0$, \ $\mathcal C\gets\{F:(F,R(F))\in\mathcal Q_0\}$, \ $n_{\mathrm{stall}}\gets0$
\State Initialize $\theta$ and perform $S_0$ updates of Eq.~\ref{eq:variational_diffusion}
on minibatches from $\mathcal Q_0$
\State $\theta_{\mathrm{ref}}\gets\theta$ and freeze $\theta_{\mathrm{ref}}$
\While{$|\mathcal A|<B$ \textbf{and} $n_{\mathrm{stall}}<N_{\mathrm{stall}}$}
    \State $a\gets|\mathcal A|$
    \State Sample $N_{\mathrm c}$ pools from $P_{\mathcal Q}$
    (Eq.~\ref{eq:buffer_sampling}) to form $\mathcal B_{\mathrm c}$
    \For{each $F\in\mathcal B_{\mathrm c}$}
        \State Corrupt $F$ to level $t_0=\max\{1,\operatorname{round}(\alpha T)\}$
        by Eq.~\ref{eq:hier_corruption}
        \State Reconstruct $F'$ from $t_0$ to $0$ with $r_\theta$ conditioned on
        Eq.~\ref{eq:pool_context}
        \State Canonicalize $F'$ and replace repeated members by random formulas
        \State $\mathcal U\gets$ canonical members of $F'$ that are not in $\mathcal A$
        \If{$|\mathcal A|+|\mathcal U|\le B$} \Comment{otherwise discard $F'$}
            \State Submit $\mathcal U$ to the evaluator and set $\mathcal A\gets\mathcal A\cup\mathcal U$
            \If{$F'$ satisfies the eligibility rule of Sec.~\ref{sec:reward_guide}}
                \State Compute $R(F')$ by Eqs.~\ref{eq:pool_reward_components}
                and~\ref{eq:pool_reward}, and set $\mathcal C\gets\mathcal C\cup\{F'\}$
                \State Insert $(F',R(F'))$ into $\mathcal Q$ and keep the top
                $C_{\mathcal Q}$ canonical pools
            \EndIf
        \EndIf
    \EndFor
    \State $n_{\mathrm{stall}}\gets0$ \textbf{if} $|\mathcal A|>a$ \textbf{else} $n_{\mathrm{stall}}+1$
    \For{$s=1,\ldots,S_{\mathrm u}$}
        \State Sample $\mathcal B_{\mathcal Q}$ of $N$ pools from $P_{\mathcal Q}$
        and evaluate Eq.~\ref{eq:variational_diffusion}
        \State Sample accepted pairs from $\mathcal P_{\mathcal Q}$ and evaluate
        Eq.~\ref{eq:pool_preference}, or set $\mathcal L_{\mathrm{pref}}=0$ if none exists
        \State Update $\theta$ by one gradient step on Eq.~\ref{eq:joint_objective}
    \EndFor
\EndWhile
\State \Return $\widehat F$ selected from $\mathcal C$ by $\widehat F=\argmax_{F\in\mathcal C}U_{\mathrm{val}}(F).$
\end{algorithmic}
\end{algorithm}

\subsection{Complexity Analysis}
\label{app:complexity}
Let $h$, $L$, and $H$ denote the hidden width, maximum formula length,
and number of reconstructed regions. With fixed Transformer depth, dense
attention, full pool re-encoding after each insertion, and key--value
caching, reconstructing one region costs
$C_{\mathrm r}=O((KL^2+K^2+L^2)h+(KL+K+L)h^2)$. Thus, forward
corruption and reverse generation cost $O(KL)$ and
$O(HC_{\mathrm r})$, respectively, with the same per-target order during
training. For $n=\sum_d N_d$ observations, pool evaluation costs
$O(K\sum_d N_d\log N_d+nK^2+K^3)$. Assuming linear-time DSL operations,
the budget $B$ bounds first-time formula execution by $O(BnL)$, while
cached signals avoid repeated execution. The elite buffer stores
$O(C_{\mathcal Q}KL)$ tokens.

\section{Supplementary Experiments}
\label{app:experiment_supplement}
This section details the experimental setup
(Appendix~\ref{app:experiment_details}) and reports additional results
(Appendix~\ref{app:additional_results}).

\subsection{Additional Experimental Details}
\label{app:experiment_details}

\subsubsection{Implementation Details}
\label{app:implementation_details}

\paragraph{Model architecture.}
The reverse model comprises a formula encoder, a pool encoder, and an
autoregressive decoder. The formula encoder is a shared two-layer
Transformer that combines token and within-formula positional embeddings
and mean-pools the non-padding states. For a queried position, its
\textsc{mask} token is replaced by \textsc{fill}, and the resulting query
copy is encoded by the same formula encoder. The pool encoder applies one
set-attention layer to the $K$ formula representations and the query
representation without row positional embeddings. Their mean-pooled
representation is combined with the diffusion-level embedding $E_t(t)$
to obtain the pool context in Eq.~\ref{eq:pool_context}. A two-layer
causal Transformer decoder conditions every token prediction on this
context and generates the reconstruction target as an RPN sequence from
\textsc{beg} to \textsc{sep}. All Transformer modules share the hidden dimension, number of
attention heads, and dropout listed in Table~\ref{tab:hyperparameter_settings}.
Grammar masking enforces the operator-arity, stack, and maximum-length
constraints during decoding.

\paragraph{Training.}
We optimize the reverse model with AdamW, initializing weight matrices
with Xavier uniform initialization and biases with zero. Before candidate
generation, the model receives $S_0$ reconstruction updates on the
initialized buffer, after which the reference model is frozen. Each
subsequent round generates $N_{\mathrm c}$ candidate pools and performs
$S_{\mathrm u}$ updates on minibatches of $N$ pools. Search rewards and
generator updates use only the 2016--2020 training period, the
2021--2022 validation period selects the configuration, checkpoint, and
final pool, and the 2023--2025 test period is used only for the final
evaluation. Table~\ref{tab:hyperparameter_settings} lists all values.

\paragraph{Static ridge combination.}
Every pool, including those returned by the baselines, is evaluated
through the same code path. Each member signal is converted on every date
to tie-aware cross-sectional ranks scaled to $[-1,1]$. The combiner solves
Eq.~\ref{eq:ridge} in closed form,
$\widehat w_F=(\widetilde X^{\top}\widetilde X/n+\gamma I)^{-1}\widetilde X^{\top}\widetilde y/n$,
without an intercept and with the penalty normalized by the pool size. The
penalty $\gamma$ is selected on a logarithmic grid by the Spearman
correlation on a chronological $80/20$ split of the training period, with
2016--2019 for fitting, 2020 for selection, and a 20-day gap between them,
after which the combiner is refitted on the full training period. The
weights are then frozen, so every validation and test date uses the same
$\widehat w_F$ and all $K$ members, without refitting or reselection. 
We use the static combiner for all
automated alpha mining methods for four reasons. It evaluates the same jointly fitted pool that the search reward optimizes. It attributes performance to the pool rather
than to an adaptive combiner, which would otherwise perform part of the
prediction for every method. Frozen weights cannot compensate for factor
decay over the three-year test period, which makes the evaluation
stricter. Finally, because all $K$ members contribute to every prediction,
the pool size directly affects the combined predictor, which the analysis
in Section~\ref{sec:pool_size} relies on.

\subsubsection{Evaluation Metrics and Backtesting Protocol}
\label{app:evaluation_protocol}

Let $\widehat{\mathbf y}_d$ and $\mathbf y_d$ denote the predicted
signals and return labels across the eligible assets on date $d$.
We define
$\rho_d=\operatorname{Pearson}(\widehat{\mathbf y}_d,\mathbf y_d)$
and
$\rho_d^{\mathrm r}
=\operatorname{Pearson}(\operatorname{rank}(\widehat{\mathbf y}_d),
\operatorname{rank}(\mathbf y_d))$,
where ties are assigned their average ranks. For a daily sequence
$\{z_d\}_{d=1}^{D}$, let
$\overline z=D^{-1}\sum_{d=1}^{D}z_d$ and $s(z)$ denote its sample
standard deviation.

\begin{itemize}[leftmargin=*,nosep]
\item \textbf{Information coefficient (IC).}
$\mathrm{IC}=|\overline{\rho}|$ measures the average daily
cross-sectional linear correlation between predictions and returns.

\item \textbf{Information coefficient information ratio (ICIR).}
$\mathrm{ICIR}=\overline{\rho}/s(\rho)$ measures the temporal stability
of IC.

\item \textbf{Rank information coefficient (RIC).}
$\mathrm{RIC}=|\overline{\rho^{\mathrm r}}|$ measures the average daily
cross-sectional rank correlation.

\item \textbf{Rank information coefficient information ratio (RICIR).}
$\mathrm{RICIR}=\overline{\rho^{\mathrm r}}/s(\rho^{\mathrm r})$ measures
the temporal stability of RIC.

\end{itemize}

Let $\Delta\Pi_t$ be the realized one-day profit or loss of all
cohorts active on trading day $t$, net of that day's transaction
costs, and let $V_0$ be the fixed initial capital. Define
$r_t=\Delta\Pi_t/V_0$ and
$W_t=1+\sum_{\tau=1}^{t}r_\tau$, with $W_0=1$.
For the benchmark, define $b_t$ from its contemporaneous one-day
profit or loss under the same fixed-capital convention, and let
$W_t^{\mathrm{bench}}=1+\sum_{\tau=1}^{t}b_\tau$.
Write $e_t=r_t-b_t$. There are $D$ test trading days and $Y$
calendar years in the test period. For a daily series $z_t$,
$\bar z$ and $s(z)$ denote its sample mean and standard deviation.

\begin{itemize}[leftmargin=*,nosep]
  \item \textbf{Annualized return (AR).}
  $\mathrm{AR}=(W_D-1)/Y$.

  \item \textbf{Sharpe ratio (SR).}
  $\mathrm{SR}=\sqrt{252}\,\bar r/s(r)$,
  with the risk-free rate set to zero.

  \item \textbf{Maximum drawdown (MDD).}
  $\mathrm{MDD}=
  \max_{1\leq t\leq D}
  \left(1-\frac{W_t}{\max_{0\leq j\leq t}W_j}\right)$.

  \item \textbf{Return information ratio (RETIR).}
  $\mathrm{RETIR}=\bar r/s(r)$.

  \item \textbf{Annualized excess return (AER).}
  $\mathrm{AER}=(W_D-W_D^{\mathrm{bench}})/Y
  =Y^{-1}\sum_{t=1}^{D}e_t$.

  \item \textbf{Information ratio (IR).}
  $\mathrm{IR}=\sqrt{252}\,\bar e/s(e)$.
\end{itemize}

\highlight{Backtesting protocol.}
For CSI300, CSI500, and S\&P500, stocks are ranked by the predicted
signal on each formation date. We form an equal-weighted long cohort
from the top 20\%, allocate one twentieth of the fixed initial capital
$V_0$ to it, and hold it for 20 trading days. A new cohort is formed
each day; thus, the portfolio's realized profit or loss on day $t$
is obtained by marking all cohorts active on that day to market.
A cohort begins earning returns only after its formation trade.

Let $G_t$ be the one-day mark-to-market profit or loss of all active
cohorts, and let $B_t$ and $S_t$ be the notionals bought and sold on
day $t$. We compute
\[
  \Delta\Pi_t
  = G_t-c_{\mathrm{buy}}B_t-c_{\mathrm{sell}}S_t,
  \qquad
  r_t=\frac{\Delta\Pi_t}{V_0},
  \qquad
  W_t=1+\sum_{\tau\leq t}r_\tau,
\]
where $c_{\mathrm{buy}}=0.0005$ and
$c_{\mathrm{sell}}=0.0015$. All return-based portfolio metrics and
cumulative-return curves are computed from this realized daily
series. The 20-trading-day forward return is used as the prediction
target. The corresponding benchmarks are SH000300, SH000905, and
\texttt{\string^GSPC}. Benchmark-relative metrics compare strategy
and benchmark gains on the same trading dates under the same capital
convention. The investable universe is constructed point-in-time:
on each date, it contains only constituents effective on that date.
Features and portfolio decisions use only information available by
the corresponding date. Samples whose 20-day return labels cross a
training, validation, or test boundary are excluded.

\subsubsection{Baseline Details}
\label{app:baselines}
We compare \shortname{} with nine baselines covering fixed factor
libraries, predictive models, and formula-discovery methods:

\begin{itemize}[leftmargin=*,nosep]
\item \textbf{Alpha158}~\citep{yang2020qlib} is a predefined library of
158 technical factors constructed from normalized OHLCV features over
multiple lookback windows. The factor matrix is fitted to the return
target using the common combination model.

\item \textbf{MLP}~\citep{murtagh1991multilayer} is a feed-forward neural network
that maps tabular features to return targets, capturing nonlinear
interactions. It provides a generic nonlinear baseline but offers
limited interpretability.

\item \textbf{XGBoost}~\citep{chen2016xgboost} uses boosted regression
trees with second-order optimization, shrinkage, feature subsampling,
and explicit regularization to model tabular financial features.

\item \textbf{GP}~\citep{chen2021empirical} evolves symbolic expressions through
selection, crossover, and subtree mutation. Candidate formulas are
evaluated by their predictive performance, and the selected formulas
form the final alpha pool.

\item \textbf{AlphaGen}~\citep{yu2023generating} formulates formula
construction as a sequential decision process and trains a masked PPO
policy. Its reward measures the improvement contributed by a generated
formula to a jointly optimized alpha pool.

\item \textbf{AlphaQCM}~\citep{zhu2025alphaqcm} extends reinforcement
learning-based alpha discovery with distributional value estimation.
The estimated return distribution and uncertainty guide exploration
under sparse and non-stationary pool rewards.

\item \textbf{AlphaForge} is a generative--predictive method for formulaic alpha discovery, which couples formula generation with predictive modeling to discover alphas. (Its $K$ selected formulas are evaluated with the shared static ridge combiner.)

\item \textbf{AlphaAgent}~\citep{tang2025alphaagent} uses an
LLM-based multi-agent workflow. Its Idea Agent proposes financial
hypotheses, the Factor Agent translates them into executable formulas,
and the Eval Agent provides empirical feedback for iterative revision.

\item \textbf{AlphaSAGE}~\citep{chen2026alphasage} represents formula
ASTs with a relational graph convolutional network and trains a
GFlowNet to sample diverse high-reward formulas. Its reward combines
predictive performance, structural alignment, and novelty.
\end{itemize}

\noindent\textbf{Shared comparison protocol.}
All baselines use the same point-in-time stock universes, chronological
splits, and 20-day forward return targets described in
Section~\ref{sec:experimental_setup}. MLP and XGBoost receive the same
Alpha158 features. Formula-discovery methods use the same primitive
market fields, formula length limit, and executable-expression
constraints. Each formula-mining method returns $K$ formulas, which are
evaluated with the static ridge combiner in
Appendix~\ref{app:implementation_details} and the backtesting protocol in
Appendix~\ref{app:evaluation_protocol}. All methods follow the same
validation-based selection.

\noindent\textbf{Search budget.}
We compare the non-LLM formula-search methods (GP, AlphaGen,
AlphaQCM, AlphaForge, AlphaSAGE, and \shortname{}) under a common
target budget of $B=40{,}000$ formula evaluations. A formula
evaluation is a call to a method's evaluator that attempts to compute
a candidate formula's signal over the historical panel. Reusing a
cached signal without a new evaluator call does not consume this
budget. The budget controls formula-evaluation work; it does not
equate optimizer updates, generated tokens, or total runtime across
methods. It is enforced at each method's evaluation call as follows:

\begin{itemize}[leftmargin=*,nosep]
\item \textbf{AlphaGen and AlphaQCM} share the same environment and
alpha pool, and training stops once the pool has evaluated
$40{,}000$ formulas.

\item \textbf{AlphaSAGE} evaluates one formula per GFlowNet episode
and runs $40{,}000$ episodes.

\item \textbf{AlphaForge} evaluates candidates in batches of $256$
and stops after the batch that reaches $40{,}000$ evaluations.
It may therefore exceed $B$ by fewer than $256$ evaluations
($<0.64\%$).

\item \textbf{GP} evolves a population of $1{,}000$ formulas for
$40$ generations and evaluates
$1{,}000\times40=40{,}000$ formulas.

\item \textbf{\shortname{}} charges each distinct canonical formula
once, on its first submission, including initialization, duplicate
replacement, and failed submissions. A candidate whose unseen
formulas would exceed the remaining budget is discarded.
\end{itemize}

For Figure~\ref{fig:efficiency_analysis}, we additionally display
$5n_{\mathrm{eval}}$ as an evaluation-equivalent token-action scale,
where $n_{\mathrm{eval}}$ is the number of formula evaluations.
The factor of five corresponds to the median formula length of four
tokens plus the terminating \textsc{sep} action. This is only a
rescaling of the figure's horizontal axis; it is not the actual
number of training steps performed by every method and is not used
as a stopping rule.

For AlphaPADI, each distinct canonical formula is charged on its first submission, including failed submissions, while cached repeats are not charged. The other methods stop according to the evaluator calls or their corresponding native units specified above. Thus, the target budget does not imply equal numbers of distinct formulas, scored pools, validation evaluations, optimizer updates, or equal runtime.

\noindent\textbf{Implementations.}
For AlphaGen\footnote{\url{https://github.com/RL-MLDM/alphagen}},
AlphaQCM\footnote{\url{https://github.com/ZhuZhouFan/AlphaQCM}},
AlphaForge\footnote{\url{https://github.com/DulyHao/AlphaForge}},
AlphaAgent\footnote{\url{https://github.com/RndmVariableQ/AlphaAgent}}, and
AlphaSAGE\footnote{\url{https://github.com/BerkinChen/AlphaSAGE}}, we
follow the default settings of their open-source projects, apart from
the shared protocol and search budget above. For Alpha158, MLP, XGBoost,
and GP, we follow the settings in the AlphaForge repository.

\subsubsection{Hyperparameter Settings}
\label{app:hyper_settings}
Baselines use the settings in Appendix~\ref{app:baselines}, with only the
pool size $K\in\{20,50\}$ selected on the validation split. For
\shortname{}, configurations are selected by the mean validation utility
across five runs and fixed before test evaluation. Shared
hyperparameters are searched over identical ranges wherever applicable,
and parameters specific to a model are tuned within that model.
All experiments are repeated five
times, and we report the mean performance across the five runs.
Statistical significance is assessed against the strongest baseline,
with $p<0.05$ considered statistically significant. 
Table~\ref{tab:hyperparameter_settings} summarizes the settings of \shortname{}.

\begin{table}[!htbp]
\centering
\caption{Hyperparameters of \shortname{}. Sets are search spaces selected on the validation split.}
\label{tab:hyperparameter_settings}
\small
\renewcommand{\arraystretch}{1.08}
\begin{tabular}{@{}llc@{}}
\toprule
\textbf{Hyperparameter} & \textbf{Symbol} & \textbf{Value} \\
\midrule
\multicolumn{3}{@{}l}{\textit{Search}} \\
Formula budget & $B$ & $40{,}000$ \\
Maximum formula length & $L_{\max}$ & $20$ \\
Pool size & $K$ & $\{20,50\}$ \\
Initial pools, buffer capacity & $M$, $C_{\mathcal Q}$ & $64$, $128$ \\
Proposal batch size & $N_{\mathrm c}$ & $16$ \\
Buffer sampling temperature & $T_{\mathrm r}$ & $0.5$ \\
\multicolumn{3}{@{}l}{\textit{Diffusion and objectives}} \\
Diffusion levels, edit strength, start level & $T$, $\alpha$, $t_0$ & $32$, $0.5$, $16$ \\
Diversity weight & $w_{\mathrm{inner}}$ & $\{1,2,5,10\}$ \\
Reconstruction weight & $\lambda_{\mathrm{VB}}$ & $\{0.1,0.3,0.5,0.7\}$ \\
Preference strength, margin & $\beta$, $m_{\mathrm{dpo}}$ & $1.0$, $0.5$ \\
\multicolumn{3}{@{}l}{\textit{Model and optimization}} \\
Hidden dimension, attention heads, dropout & -- & $128$, $4$, $0.1$ \\
Formula, set, and decoder layers & -- & $2$, $1$, $2$ \\
Optimizer, weight decay, gradient clipping & -- & AdamW, $10^{-4}$, $1.0$ \\
Learning rate & -- & $\{2{\times}10^{-5},3{\times}10^{-4},10^{-3}\}$ \\
Minibatch size, initial updates, updates per round & $N$, $S_0$, $S_{\mathrm u}$ & $16$, $16$, $8$ \\
\bottomrule
\end{tabular}
\end{table}

\begin{figure}[!t]
    \centering
    \includegraphics[width=\textwidth]{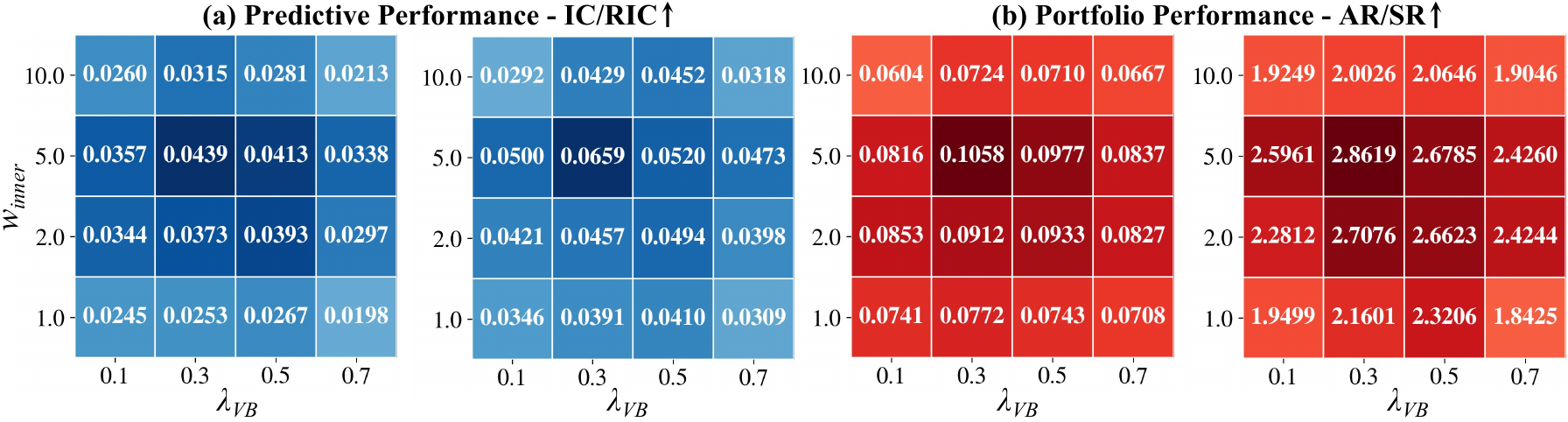}
    \caption{Joint sensitivity to $w_{\mathrm{inner}}$ (rows) and $\lambda_{\mathrm{VB}}$ (columns) on CSI300. Panel (a) shows IC and RIC, and panel (b) shows AR and SR. Annotations give metric values, and colors are normalized separately within each metric. Higher values are better for all four metrics.}
    \label{fig:hyperparameter_analysis}
\end{figure}

\begin{figure}[!t]
    \centering
    \includegraphics[width=\textwidth]{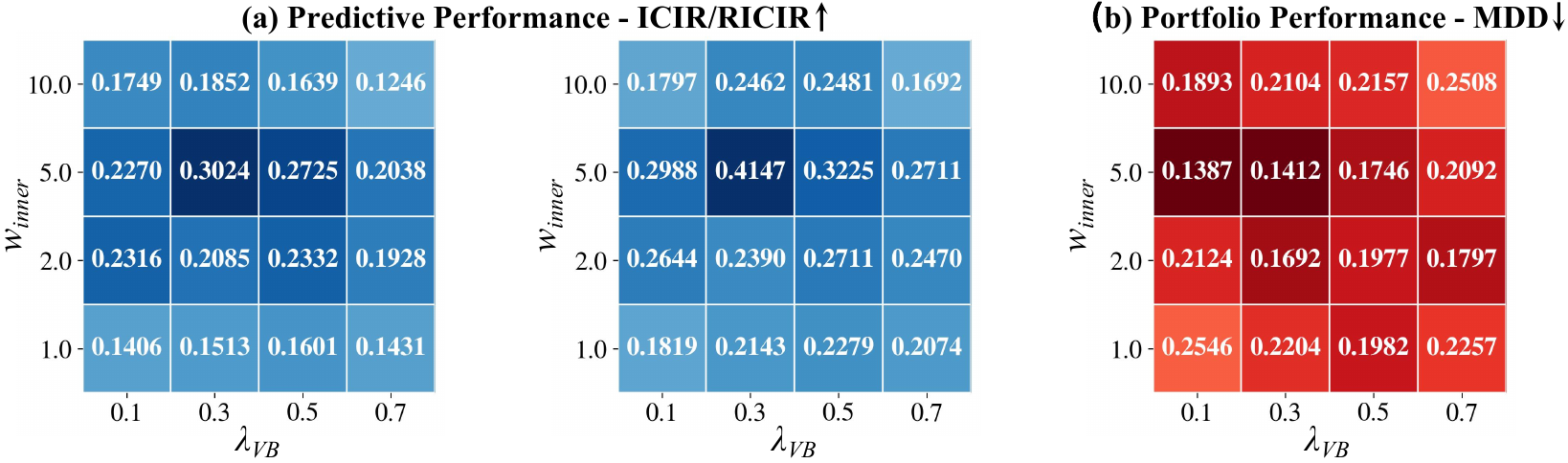}
    \caption{Joint sensitivity to $w_{\mathrm{inner}}$ (rows) and $\lambda_{\mathrm{VB}}$ (columns) on CSI300. Panel (a) shows ICIR and RICIR, and panel (b) shows MDD. Annotations give metric values, and colors are normalized separately within each metric. Higher ICIR and RICIR and lower MDD are better.}
    \label{fig:hyperparameter_analysis_app}
\end{figure}

\subsection{Additional Experimental Results}
\label{app:additional_results}

\begin{figure}[!t]
    \centering
    \includegraphics[width=0.9\textwidth]{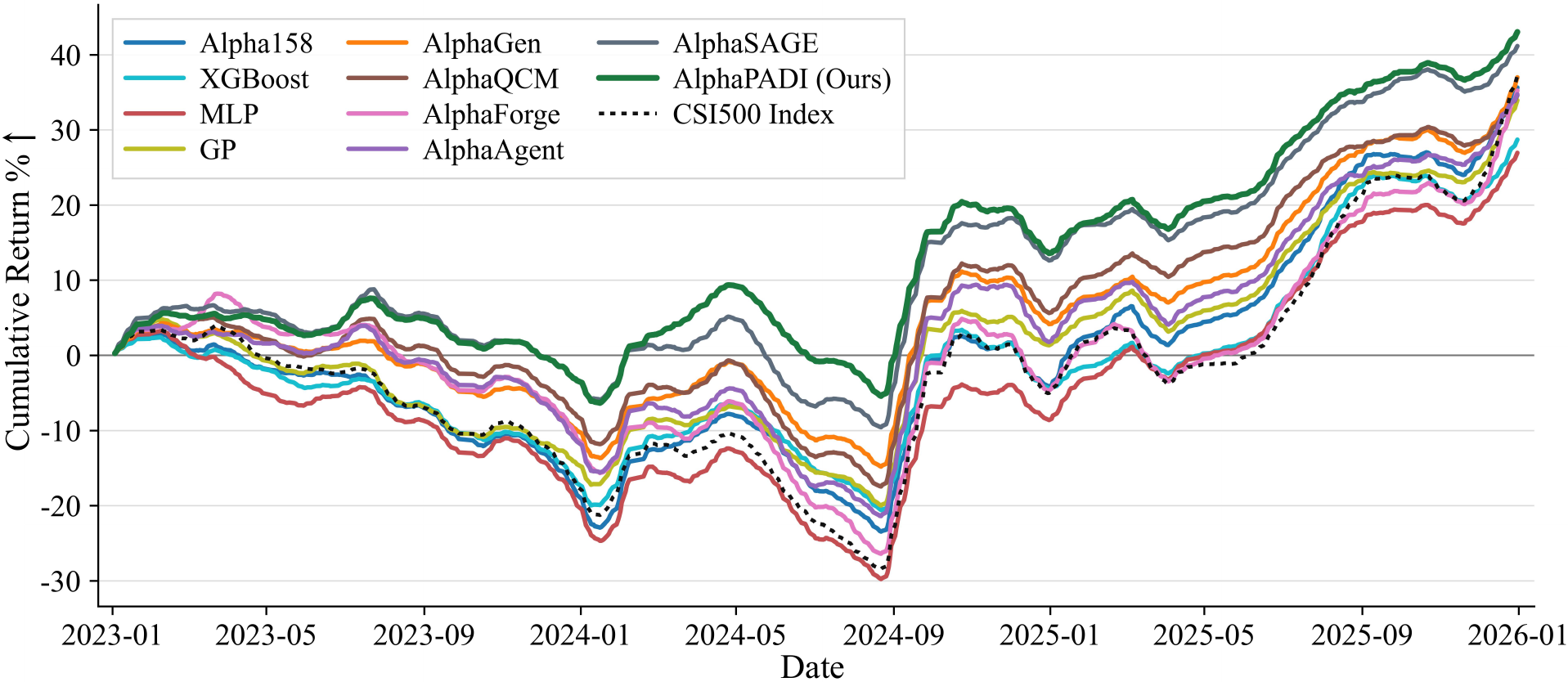}
    \caption{Cumulative returns on CSI500 (2023--2025): \shortname{}, baseline methods, and the CSI500 index.}
    \label{fig:cr_plot_csi500}
\end{figure}

\begin{figure}[!t]
    \centering
    \includegraphics[width=0.9\textwidth]{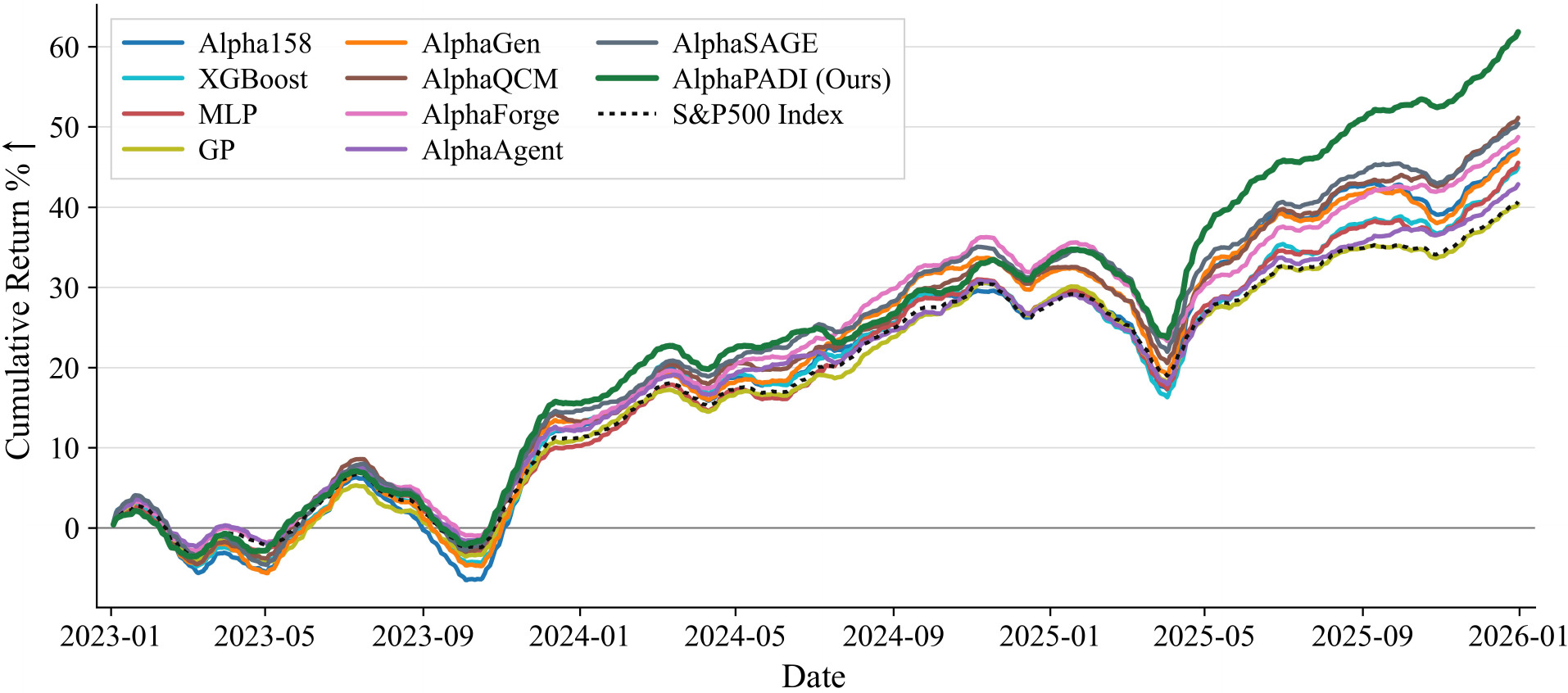}
    \caption{Cumulative returns on S\&P500 (2023--2025): \shortname{}, baseline methods, and the S\&P500 index.}
    \label{fig:cr_plot_sp500}
\end{figure}

\subsubsection{Hyperparameter Analyses}
\label{app:hyperparameter_analyses}
We vary only $w_{\mathrm{inner}}$ and $\lambda_{\mathrm{VB}}$, keeping all other experimental settings and the checkpoint-selection protocol the same as in the main experiments. Figures~\ref{fig:hyperparameter_analysis} and~\ref{fig:hyperparameter_analysis_app} show joint sensitivity to the inner diversity reward weight $w_{\mathrm{inner}}$ and reconstruction-loss weight $\lambda_{\mathrm{VB}}$ on CSI300. Among the displayed configurations, $(w_{\mathrm{inner}},\lambda_{\mathrm{VB}})=(5.0,0.3)$ has the highest IC, RIC, AR, and SR. Moderate diversity weights generally perform better than the smallest and largest tested values, while increasing the reconstruction weight to $0.7$ reduces performance relative to the moderate range. These patterns describe sensitivity within the tested grid.

Figure~\ref{fig:hyperparameter_analysis} complements the IC, RIC, AR, and SR grid in Figure~\ref{fig:hyperparameter_analysis_app} with ICIR, RICIR, and MDD on CSI300. The setting $(w_{\mathrm{inner}},\lambda_{\mathrm{VB}})=(5.0,0.3)$ attains the highest ICIR ($0.3024$) and RICIR ($0.4147$) in the displayed grid. Holding $w_{\mathrm{inner}}=5.0$, increasing $\lambda_{\mathrm{VB}}$ to $0.7$ lowers these metrics to $0.2038$ and $0.2711$, respectively, and increases MDD from $0.1412$ to $0.2092$. The minimum MDD, however, is $0.1387$ at $(5.0,0.1)$, compared with $0.1412$ at $(5.0,0.3)$. Thus, the configuration with the strongest predictive information ratios does not also minimize drawdown. The grid therefore reveals a trade-off between predictive stability and drawdown within the evaluated configurations.

\subsubsection{Backtesting Results}
\label{app:backtesting}
Figures~\ref{fig:cr_plot_csi500} and~\ref{fig:cr_plot_sp500} extend the
CSI300 backtest in the main text to CSI500 and S\&P500. On CSI500,
\shortname{} has a shallower cumulative-return trough than the baselines
during the decline preceding September 2024, participates in the
subsequent recovery, and ends above the competing strategies, with
AlphaSAGE providing the closest final comparison. On S\&P500, the methods
track one another more closely through much of 2023--2024, and the
separation becomes pronounced after the decline in early 2025, when
\shortname{} recovers and finishes above the baselines and the index.
Together, the two plots show when the differences in
Table~\ref{tab:main_results} emerge.

\newcommand{\AlphaDashedRule}{%
  \noalign{%
    \vskip 2pt
    \hbox to \textwidth{%
      \leaders\hbox{\rule{3pt}{0.4pt}\hskip 2pt}\hfill
    }%
    \vskip 2pt
  }%
}

\begin{table}[t]
\centering
\caption{
Comparison of Alpha158 initialization, Alpha360 initialization, and training from scratch on CSI300, including the original factor-library baselines.
Higher values are better except for MDD.
Best results are \textbf{bolded}, and second-best results are
\underline{underlined}.
}
\label{tab:initialization}

\footnotesize
\rmfamily
\setlength{\tabcolsep}{3.0pt}
\renewcommand{\arraystretch}{1.15}

\begin{tabularx}{\textwidth}{@{}
  >{\centering\arraybackslash}m{1.45cm}
  >{\centering\arraybackslash}m{2.05cm}
  *{7}{>{\centering\arraybackslash}X}
@{}}
\toprule
\multirow[c]{2}{1.45cm}{\centering\textbf{Method}}
&
\multirow[c]{2}{2.05cm}{\centering\textbf{Initialization}}
&
\multicolumn{4}{c}{\textbf{Predictive Performance}}
&
\multicolumn{3}{c}{\textbf{Portfolio Performance}}
\\
\cmidrule(lr){3-6}
\cmidrule(lr){7-9}
&
&
\textbf{IC$\uparrow$}
&
\textbf{ICIR$\uparrow$}
&
\textbf{RIC$\uparrow$}
&
\textbf{RICIR$\uparrow$}
&
\textbf{AR$\uparrow$}
&
\textbf{SR$\uparrow$}
&
\textbf{MDD$\downarrow$}
\\
\midrule

Alpha158
& --
& 0.0266 & 0.1631 & 0.0250 & 0.1389
& 0.0751 & 2.1532 & 0.2274
\\

AlphaPADI
& Alpha158
& \textbf{0.0462} & 0.2497 & \underline{0.0632} & 0.3144
& 0.0952 & \textbf{3.0287} & \underline{0.1582}
\\
\AlphaDashedRule

Alpha360
& --
& 0.0067 & 0.0326 & 0.0131 & 0.0611
& 0.0538 & 0.9706 & 0.4220
\\

AlphaPADI
& Alpha360
& 0.0401 & \underline{0.2539} & 0.0541 & \underline{0.3297}
& \underline{0.0967} & 2.4779 & 0.1687
\\
\AlphaDashedRule

AlphaPADI
& \mbox{From Scratch}
& \underline{0.0439} & \textbf{0.3024}
& \textbf{0.0659} & \textbf{0.4147}
& \textbf{0.1058} & \underline{2.8619} & \textbf{0.1412}
\\
\bottomrule
\end{tabularx}
\end{table}

\begin{figure}[t]
    \centering
    \includegraphics[width=\textwidth]{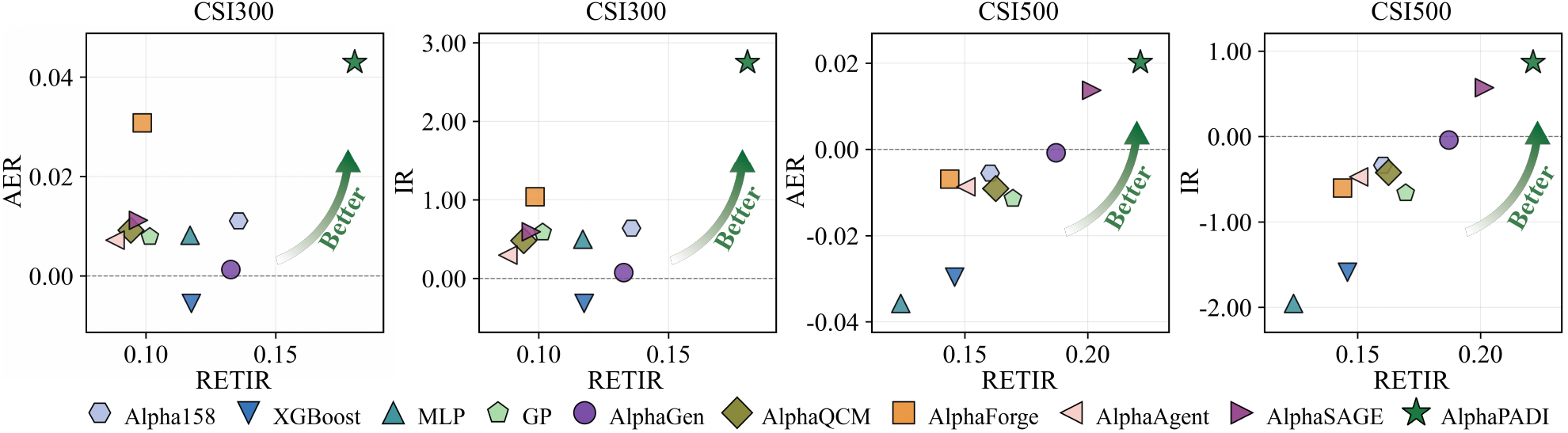}
    \caption{Additional portfolio metrics on CSI300 and CSI500: return information ratio (RETIR), annualized excess return (AER), and information ratio (IR).}
    \label{fig:other_metrics_main}
\end{figure}

\begin{figure}[!t]
    \centering
    \includegraphics[width=\textwidth]{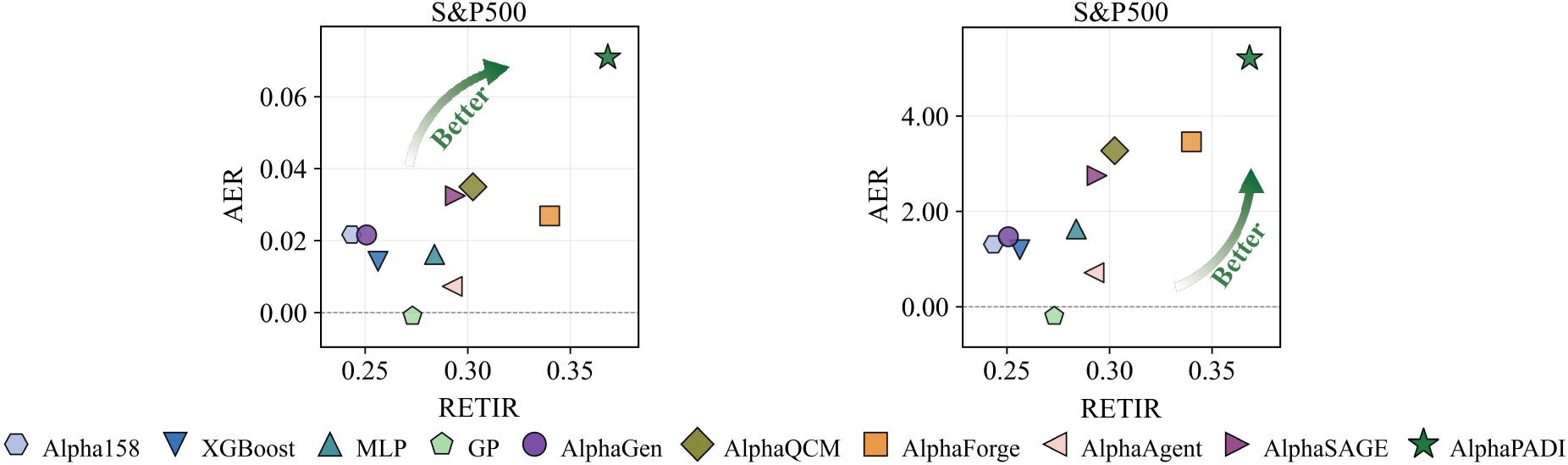}
    \caption{Additional portfolio comparisons on S\&P500. The left panel plots AER against RETIR.}
    \label{fig:other_metrics_app}
\end{figure}

\begin{figure}[!htbp]
    \centering
    \includegraphics[width=\textwidth]{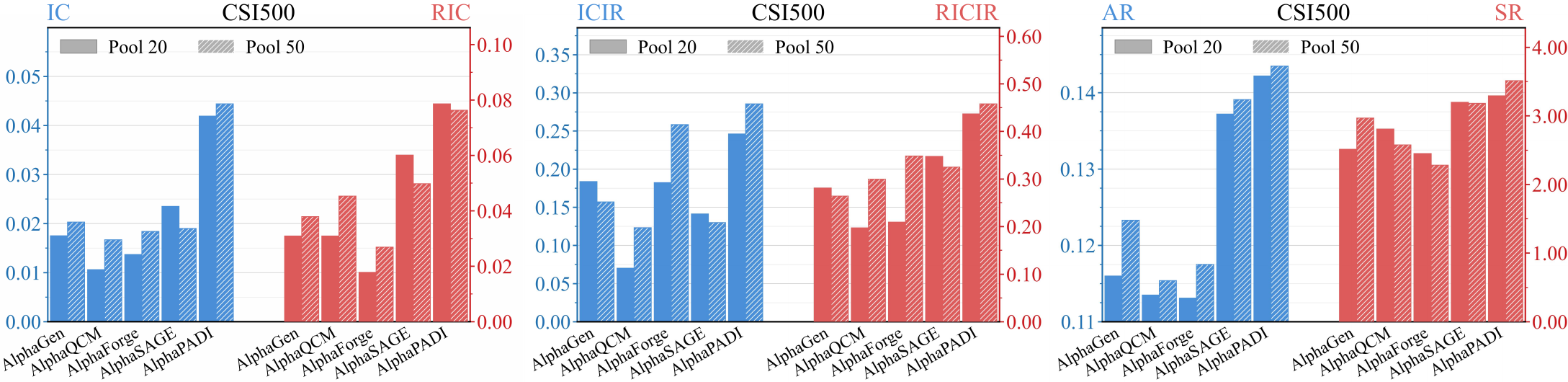}
    \caption{Performance of five alpha-mining methods with pool sizes $K=20$ (solid) and $K=50$ (hatched) on CSI500. Blue and red bars use the corresponding left and right metric axes.}
    \label{fig:pool_size_csi500}
\end{figure}

\begin{figure}[!htbp]
    \centering
    \includegraphics[width=\textwidth]{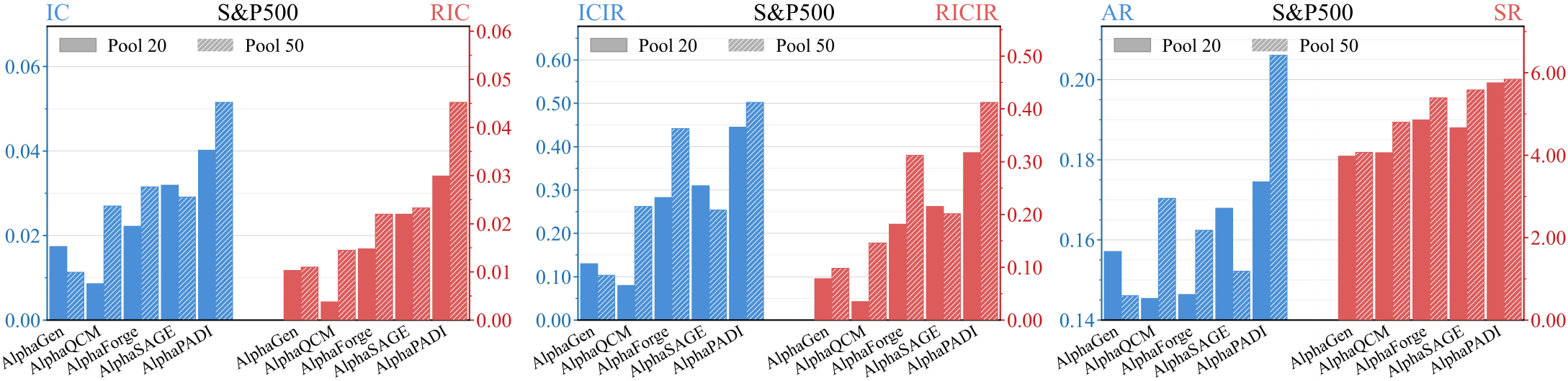}
    \caption{Performance of five alpha-mining methods with pool sizes $K=20$ (solid) and $K=50$ (hatched) on S\&P500. Blue and red bars use the corresponding left and right metric axes.}
    \label{fig:pool_size_sp500}
\end{figure}

\begin{figure}[!t]
    \centering
    \includegraphics[width=0.92\linewidth]{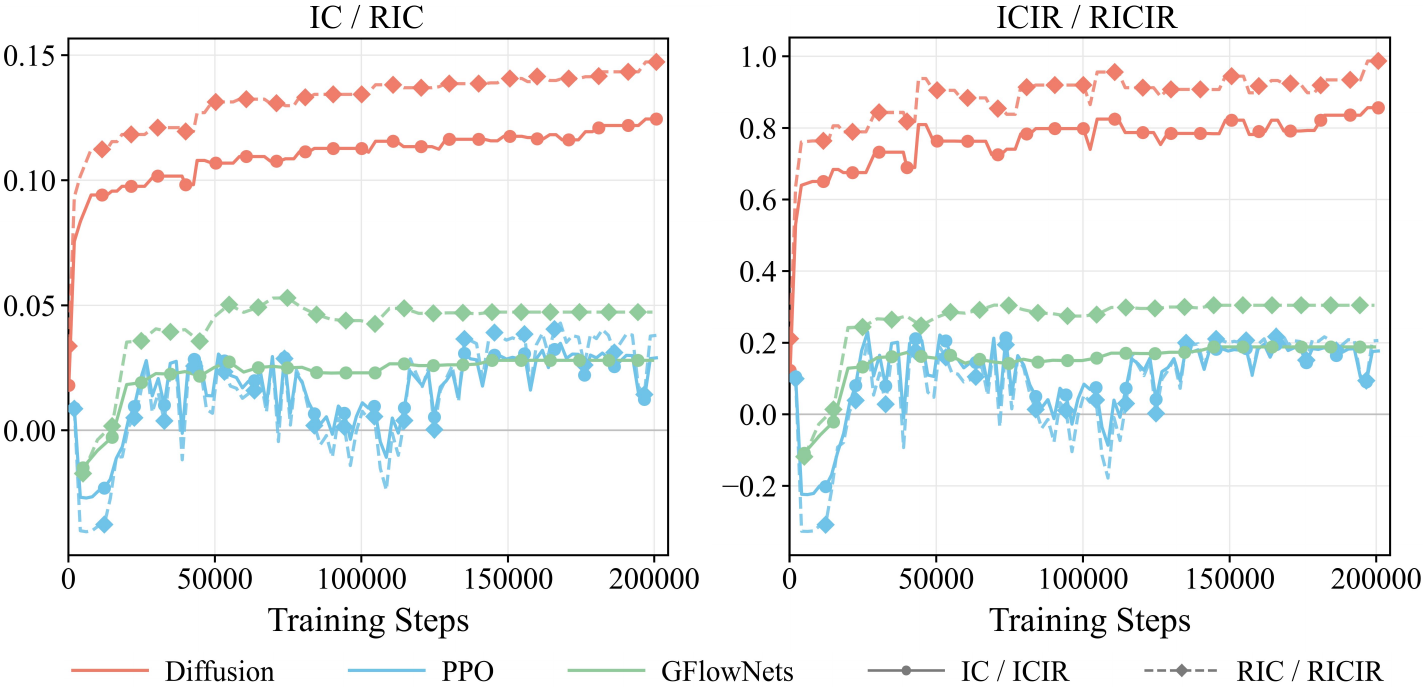}
    \caption{Search dynamics on CSI300: training-split IC/RIC and ICIR/RICIR
of the pool each method would return during search. The horizontal
coordinate labeled Training Steps in the plot is defined as
$5n_{\mathrm{eval}}$, where $n_{\mathrm{eval}}$ is the number of formula
evaluations. This fixed rescaling is used only for display and does not
represent optimizer updates, generated actions, or elapsed training time.}
    \label{fig:efficiency_analysis}
\end{figure}

\subsubsection{Cold-start Analysis}
\label{app:cold_start}
To evaluate whether \shortname{} depends on a predefined factor library, we compare initialization with Alpha158 and Alpha360 against training from scratch
on CSI300. The from-scratch variant ranks first on five of the seven reported metrics and second on the remaining two.
Specifically, it obtains the best ICIR (0.3024), RIC (0.0659), RICIR (0.4147), AR (0.1058), and MDD (0.1412), while achieving the second-best IC (0.0439) and
SR (2.8619). Initialization remains beneficial for particular objectives: Alpha158 initialization produces the highest IC of 0.0462 and SR of 3.0287.
Refinement also improves the reported metrics relative to the two original factor sets. With Alpha360 initialization, for example, it increases IC from 0.0067 to 0.0401
and AR from 0.0538 to 0.0967, while reducing MDD from 0.4220 to 0.1687. For the two libraries tested, refinement improves performance over the corresponding library baselines. The from-scratch results further show that a predefined factor library is not required to obtain competitive performance on CSI300.

\subsubsection{Performance on Other Metrics}
\label{app:other_metrics}

Figure~\ref{fig:other_metrics_main} presents a complementary evaluation using portfolio-level metrics: return information ratio (RETIR), annualized excess return (AER), and information ratio (IR) on the CSI300 and CSI500 stock universes. These indicators complement the main results with additional measures of portfolio returns and benchmark-relative performance. \shortname{} has the highest plotted values on both axes in each comparison. On CSI300, it achieves a RETIR of 0.1803, an AER of 0.0430, and an IR of 2.7520, outperforming AlphaForge by approximately 0.01 in AER and 1.7 in IR. On CSI500, \shortname{} obtains a RETIR of 0.2213, an AER of 0.0202, and an IR of 0.8696. Relative to AlphaSAGE, the strongest competing method on this benchmark, these results correspond to improvements of approximately 10.5\% in RETIR, 46.8\% in AER, and 51.5\% in IR. Notably, while several baselines produce near-zero or negative excess performance on CSI500, \shortname{} remains positive on these reported metrics. These results provide additional evidence of positive benchmark-relative performance on the two evaluated stock universes.

Figure~\ref{fig:other_metrics_app} extends the additional portfolio evaluation to S\&P500. In the left panel, \shortname{} occupies the upper-right position, with both a higher RETIR and a higher AER than the displayed baselines. This complements the Chinese-universe results by showing that the return-based and excess-return advantages also appear in the U.S. universe. AlphaForge has a higher RETIR but a lower AER than AlphaQCM, illustrating that rankings can differ across metrics.

\subsubsection{Search Dynamics Under a Shared Formula-Evaluation Budget}
\label{app:efficiency}

With the budget fixed at $40{,}000$ formula evaluations, the search-dynamics
plot tracks the training-split quality of the pool each method would return
at a given evaluation count. Its horizontal coordinate is $5n_{\mathrm{eval}}$
as defined in the caption; comparisons along this axis concern formula
evaluations rather than actual training steps or wall-clock efficiency.
Because the curves are measured on the training split, they describe search
progress rather than generalization, which is evaluated in Table~1.

AlphaPADI reaches an IC of approximately $0.095$ after about $2{,}000$
formula evaluations (displayed at $10^4$ on the rescaled horizontal axis).
At the end of the budget, its IC and RIC reach $0.1245$ and $0.1473$,
respectively, whereas the PPO and GFlowNet curves end at an IC of
approximately $0.03$ or lower.

\subsubsection{Effect of Pool Size}
\label{app:pool_size}

Figures~\ref{fig:pool_size_csi500} and~\ref{fig:pool_size_sp500} extend the pool-size comparison to CSI500 and S\&P500. They compare pools of $K=20$ and $50$ for AlphaGen, AlphaQCM, AlphaForge, AlphaSAGE, and \shortname{}. Increasing pool size does not uniformly improve every method or metric. On CSI500, for example, \shortname{} improves IC and ICIR with the larger pool, while its RIC decreases slightly. On S\&P500, its larger pool improves all six displayed metrics, including IC, RIC, AR, and SR. The differences across methods show that additional pool capacity alone does not guarantee better performance.

\end{document}